\documentclass[11pt]{article}

\input{Style.sty}

\usepackage{graphicx}
\graphicspath{{fig/}}

\title{Sparse Quantum State Preparation with Sublinear T-Count}
\author{Jingquan Luo\thanks{luojq25@mail2.sysu.edu.cn} }
\author{Lvzhou Li\thanks{Corresponding author: lilvzh@mail.sysu.edu.cn}}
\affil{Institute of Quantum Computing  and Software, School of Computer Science and Engineering, Sun Yat-sen University, Guangzhou 510006, China}

\author{}
\date{ }

\begin{document}
\maketitle

\begin{abstract}
We study the fault-tolerant cost of preparing sparse quantum states, measured by $T$-count in the Clifford+$T$ model. Here an $n$-qubit state is called $s$-sparse if it is supported on at most $s$ computational-basis states. For arbitrary $n$-qubit states, the optimal $T$-count is $\Theta(\sqrt{2^n\log(1/\epsilon)}+\log(1/\epsilon))$, but for $s$-sparse states the best previous upper bounds remained linear in $s$. We show that any $n$-qubit $s$-sparse state can be prepared up to error $\epsilon$ using $\widetilde{O}(\min\{s,\ n^{3/4}\sqrt{s}\}+\sqrt{s\log(1/\epsilon)}+\log(1/\epsilon))$ $T$ gates, giving the first sublinear dependence on $s$ once the support is sufficiently large. Our approach is based on a support-aware synthesis theorem for sparse Boolean functions, which may be of independent interest. We also prove that, for every $0<\epsilon\le 1/6$ and $2\le s\le 2^{n/2}$, sparse-state preparation requires $\Omega(\min\{s,\sqrt{ns}\})$ $T$ gates, showing that linear dependence on $s$ is unavoidable in the small-support regime and substantially narrowing the gap between the known upper and lower bounds within this parameter range.
\end{abstract}

\section{Introduction}
\label{sec:introduction}
Quantum state preparation is a central primitive in quantum computing. It underlies the initialization of structured superpositions in tasks such as quantum linear systems, Hamiltonian simulation, and amplitude-encoded quantum machine learning \cite{harrow2009quantum,childs2017quantum,low2019hamiltonian,kerenidis2019q}. At the same time, canonical synthesis results for arbitrary states show that the task is expensive in the worst case, because specifying a generic $n$-qubit pure state requires $\Theta(2^n)$ real parameters \cite{shende2005synthesis,plesch2011quantum,iten2016quantum}. This makes state preparation especially interesting when the target state has additional structure that can be translated into a circuit advantage.

In the fault-tolerant setting, the relevant resource is usually not the total number of elementary gates but the number of non-Clifford gates, especially $T$ gates. The Clifford+$T$ gate set is universal, but it is also discrete, so exact synthesis is unavailable for generic target amplitudes. We therefore study approximate state preparation in the unitary Clifford+$T$ model~\cite{nielsen2010quantum}. For density operators $\rho$ and $\sigma$, let $D_{\mathrm{tr}}(\rho,\sigma):=\frac12\|\rho-\sigma\|_1$. For pure states $\ket{\psi}$ and $\ket{\phi}$, let $D_{\mathrm{tr}}(\ket{\psi},\ket{\phi}):=D_{\mathrm{tr}}(\ket{\psi}\!\bra{\psi},\ket{\phi}\!\bra{\phi})=\sqrt{1-\left|\braket{\psi|\phi}\right|^2}$. We say that a unitary Clifford+$T$ circuit $\epsilon$-approximately prepares an $n$-qubit state $\ket{\psi}$ if, for some $a\ge 0$, starting from $\ket{0^{n+a}}$ it outputs an $(n+a)$-qubit pure state $\ket{\phi}$ such that
\begin{align}
D_{\mathrm{tr}}\!\left(\ket{\phi},\ket{\psi}\otimes \ket{0^a}\right)\le \epsilon.
\end{align}
Within this model, Clifford operations are comparatively inexpensive to implement fault tolerantly, whereas $T$ gates are costly because they rely on magic-state injection and distillation \cite{litinski2019game, bravyi2012magic}. Throughout this paper, we therefore use $T$-count as the primary complexity measure for state preparation.

Recent work has clarified the $T$ complexity of general quantum state preparation. Let $N = 2^n$. Low et al. established explicit tradeoffs between ancillary workspace and $T$-count, and Gosset et al. showed that an arbitrary $n$-qubit state can be $\epsilon$-approximately prepared using $O(\sqrt{N\log(1/\epsilon)} + \log(1/\epsilon))$ $T$ gates and ancillas, and that this scaling is asymptotically optimal for general states in the stronger model of adaptive Clifford+$T$ circuits \cite{low2024trading,gosset2024quantum}. These results show that ancillary workspace can significantly reduce the cost of unstructured state preparation. However, the worst-case complexity is still exponential in $n$, because the target state is treated as completely generic.

A natural structured family is that of sparse quantum states. We call an $n$-qubit state $s$-sparse if its support in the computational basis has size at most $s$; equivalently, there exist distinct basis strings $x^{(0)},\ldots,x^{(s-1)} \in \{0,1\}^n$ and coefficients $\alpha_0,\ldots,\alpha_{s-1}\in\mathbb C$ such that
\begin{align}
|\psi\rangle = \sum_{j=0}^{s-1} \alpha_j |x^{(j)}\rangle,
\end{align}
where $\sum_{j=0}^{s-1} |\alpha_j|^2 = 1$. We denote its support by $\operatorname{supp}(\ket{\psi}) := \{x\in\{0,1\}^n : \langle x |\psi\rangle \neq 0\}$. When $s \ll 2^n$, the natural question is whether the synthesis cost is governed by the support size rather than by the ambient Hilbert-space dimension.

This question is supported by a substantial circuit-complexity literature on sparse state preparation~\cite{gleinig2021efficient, malvetti2021quantum, ramacciotti2023simple, mozafari2022efficient, de2020circuit, de2022double, zhang2022quantum, zhang2024circuit, wang2025faster, harrow2025randomized}. Earlier work gave efficient constructions for sparse states with circuit size $O(ns)$~\cite{gleinig2021efficient, malvetti2021quantum, ramacciotti2023simple, mozafari2022efficient, de2020circuit, de2022double}. Subsequent results sharpened this bound and developed refined space-time tradeoffs \cite{mao2024towards,luo2024circuit,luo2025space}. However, these works focus primarily on total gate count, depth, or space-time complexity rather than the fault-tolerant non-Clifford cost.

Only a small number of recent works speak directly to that metric~\cite{vilmart2025resource, rupprecht2026sparse}. In particular, Vilmart, Ty, and Mang~\cite{vilmart2025resource} gave a sparse-state preparation algorithm with $O(s)$ non-Clifford gates and $\max\{s-n, 0\}$ ancillas. Rupprecht and W\"olk give a sparse-state preparation method with Toffoli complexity $O(s)$ \cite{rupprecht2026sparse}; since each Toffoli admits a constant-size Clifford+$T$ implementation~\cite{nielsen2010quantum}, this yields an $O(s)$ contribution to the corresponding $T$-count. Combined with the optimal dense-state preparation bound on a $\lceil \log s \rceil$-qubit register \cite{gosset2024quantum}, this gives an $O(s + \log(1/\epsilon))$ upper bound for $\epsilon$-approximate preparation of $s$-sparse states. This leaves a basic question: can the $T$-count for preparing $s$-sparse states be made sublinear in $s$, at least in some parameter regime?

\subsection{Contribution}

Our main result is a new $T$-count upper bound that becomes sublinear in $s$ once the support reaches $\widetilde{\Omega}(n^{3/2})$, together with lower bounds showing that, for $0<\epsilon\le \frac{1}{6}$, linear dependence on $s$ is unavoidable for $s = O(n)$ and that every construction must still pay at least $\Omega(\sqrt{ns})$ for $\Omega(n) \le s \le 2^{n/2}$. Table~\ref{tab:intro-sqsp-complexity} and Figure~\ref{fig:intro-sparse-bounds} summarize this picture. We first state the constant-error form of the upper bound to highlight the dependence on $n$ and $s$.

\begin{theorem}[Informal version of \cref{thm:sqsp-main}]
\label{thm:intro-sqsp-informal}
    Any $n$-qubit $s$-sparse state can be prepared up to constant error by a Clifford+$T$ circuit starting with the all-zeros state using
    $\widetilde{O}\left(n^{3/4}\sqrt{s}\right)$
    $T$ gates and $\widetilde{O}(\sqrt{ns})$ ancillas\footnote{The $\widetilde{O}$ notation suppresses polylogarithmic factors, such as $\log n$ and $\log s$.}.
\end{theorem}

Combining Theorem~\ref{thm:intro-sqsp-informal} with the previous linear-in-$s$ construction~\cite{rupprecht2026sparse,vilmart2025resource} yields the following best known upper bound.

\begin{corollary}
\label{cor:intro-sqsp-combined}
Any $n$-qubit $s$-sparse state can be prepared up to constant error by a Clifford+$T$ circuit starting with the all-zeros state using
$\widetilde{O}(\min\{s,\ n^{3/4}\sqrt{s}\})$
$T$ gates.
\end{corollary}

Thus the previous linear-in-$s$ construction remains the better upper bound for $s = \widetilde{O}(n^{3/2})$, while our new construction gives the first sublinear dependence on $s$ beyond that threshold.

% Table~\ref{tab:intro-sqsp-complexity} summarizes the upper bounds on the $T$-count and ancilla complexity for $\epsilon$-approximate preparation of sparse quantum states.

We complement this upper bound with lower bounds showing that neither a sublinear dependence on $s$ in the small-support regime nor a complete removal of the dependence on $n$ is possible in general.
% Let $\mathcal T^{\mathrm{spar}}(n,s)$ denote the worst-case $T$-count of $\epsilon$-approximately preparing an $n$-qubit $s$-sparse state.
\begin{theorem}[Informal version of \cref{thm:sparse-state-lower-bound}]
\label{thm:intro-lower-bound-informal}
For every $0<\epsilon\le \frac{1}{6}$ and every $2 \leq s \leq 2^{n/2}$, there exists an $n$-qubit $s$-sparse state $\ket{\psi}$ such that any Clifford+$T$ circuit that prepares $\ket{\psi}$ up to error $\epsilon$ costs $\Omega(\min\{s, \sqrt{ns}\})$ $T$ gates.
\end{theorem}

In particular, the linear upper bound $O(s)$ is tight when $s = O(n)$, so no construction can achieve a sublinear dependence on $s$ in this regime. By contrast, for $\widetilde{\Omega}(n^{3/2}) \le s \le 2^{n/2}$, our new upper bound differs from the lower bound by at most a multiplicative $\widetilde{O}(n^{1/4})$ factor, showing that any sublinear dependence on $s$ in this range must inevitably come with a dependence on $n$. Figure~\ref{fig:intro-sparse-bounds} summarizes the relation among the previous upper bound, our new upper bound, and these lower bounds.

\begin{table}[htbp]
\centering
\small
\renewcommand{\arraystretch}{1.5}
\caption{Known constant-error bounds for approximate preparation of an $n$-qubit quantum state.}
\label{tab:intro-sqsp-complexity}
\begin{tabular}{|c|c|c|c|}
\hline
\textbf{Type} & \textbf{Source} & \textbf{$T$-count} & \textbf{Ancillas} \\
\hline
General state 
& \cite{gosset2024quantum} & $\theta(\sqrt{2^n})$ & $O(\sqrt{2^n})$ \\
\hline
\multirow{4}{*}{Sparse state}
& \cite{vilmart2025resource} & $O(s)$ & $\max\{s-n, 0\}$ \\
\cline{2-4}
& \cite{rupprecht2026sparse} & $O\!\left(s\right)$ & $O\!\left(\log s\right)$ \\
\cline{2-4}
& This work & $\widetilde{O}(n^{3/4}\sqrt{s})$ & $\widetilde{O}\!\left(\sqrt{ns} \right)$ \\
\cline{2-4}
& This work ($s \le 2^{n/2}$) & $\Omega(\min\{s, \sqrt{ns}\})$ & unbounded \\
\hline
\end{tabular}
\end{table}

\begin{figure}[htbp]
\centering
\includegraphics[width=0.7\textwidth]{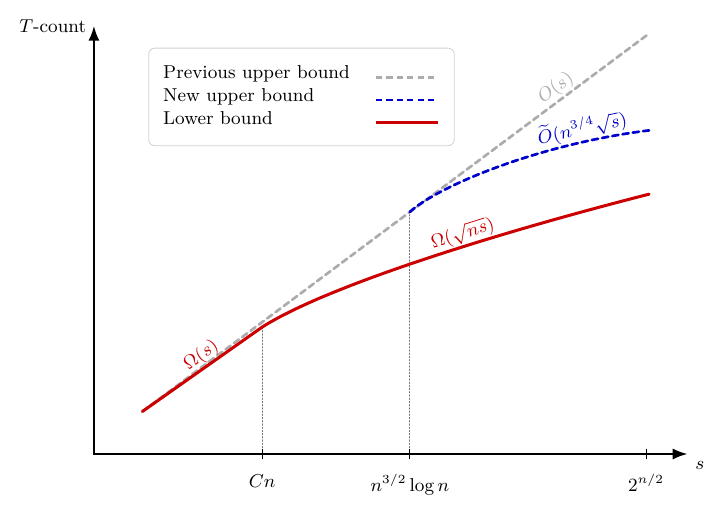}
\caption{Schematic relation between the known constant-error bounds on the worst-case $T$-count for preparing an $n$-qubit $s$-sparse state as a function of the sparsity $s$. Ignoring polylogarithmic factors, our new upper bound is $\widetilde{O}(n^{3/4}\sqrt{s})$, which crosses the previous linear upper bound around $s=n^{3/2}$ up to polylogarithmic factors. The lower bound is linear for $s = O(n)$ and becomes $\Omega(\sqrt{ns})$ for $s = \Omega(n)$ up to $s \le 2^{n/2}$.}
\label{fig:intro-sparse-bounds}
\end{figure}

\paragraph{Technique overview.}

The upper bound is based on a reduction from sparse-state preparation to sparse Boolean function synthesis. Writing $|\psi\rangle = \sum_{i=0}^{s-1}\alpha_i|x^{(i)}\rangle$, we first prepare the dense label state $\sum_{i=0}^{s-1}\alpha_i\ket{i}$ on a $\lceil \log s \rceil$-qubit index register and then coherently load the basis string $x^{(i)}$ into an $n$-qubit data register using the Boolean function $i \mapsto x^{(i)}$. The main difficulty is the final label-erasure step: we apply the zero-extension of the inverse labeling map, namely the Boolean function that maps $x^{(i)}$ back to $i$ on the support and outputs $0$ elsewhere. Unlike the first two steps, this operation acts on the full $n$-bit data register rather than on the $\lceil \log s \rceil$-qubit label register. This bottleneck motivates the following synthesis theorem for sparse Boolean functions, which may be of independent interest.

\begin{theorem}[Informal version of \cref{thm:sparse-total-boolean}]
\label{thm:intro-sparse-boolean-informal}
Let $F : \{0,1\}^n \to \{0,1\}^b$ be a Boolean function whose support
$\operatorname{supp}(F):=\{x\in\{0,1\}^n : F(x)\neq 0^b\}$
has size at most $s$. When $s \le n^{3/2}\log n$, $F$ can be implemented by a NOT/CNOT/Toffoli (NCT) circuit with $O(s+n)$ Toffoli gates \footnote{If we only require the function to be correct on the support, that is, to implement a partial Boolean function, then the number of Toffoli gates can be optimized to $O(s)$.} and ${O}(\log s)$ ancillas. When $s > n^{3/2}\log n$, $F$ can be implemented by an NCT circuit with $\widetilde{O}(n^{3/4}\sqrt{s} + n^{1/4}\sqrt{bs})$
Toffoli gates and 
$\widetilde{O}(n^{1/4}\sqrt{s} + n^{-1/4}\sqrt{bs})$
ancillas.
\end{theorem}

For the state-preparation application, the relevant output length is $b=\lceil \log s \rceil \le n$. Substituting this into the large-support regime of Theorem~\ref{thm:intro-sparse-boolean-informal} gives Theorem~\ref{thm:intro-sqsp-informal}.

A naive implementation applies one multi-controlled Toffoli gate for each supported input and therefore uses $O(ns)$ Toffoli gates. This could be reduced to $O(s+n)$ by the technique proposed in~\cite{rupprecht2026sparse}, but the dependence on $s$ remains linear. Our theorem preserves that behavior in the small-support regime and improves the large-support regime to $\widetilde{O}(n^{3/4}\sqrt{s})$ when $b\le n$.

At a high level, the large-support algorithm fixes a prefix length $r=\lceil \log(s/\lceil\sqrt{n}\rceil)\rceil$ and recursively analyzes the active $r$-bit prefixes in the current support. If only a few such prefixes are active, then we compress them by injectively relabeling the active prefixes, which reduces the input length without changing the support size. If many prefixes are active, then we peel off one supported input from each active prefix; this peeled slice can be implemented cheaply, while the remaining recursive instance loses at least $2^{r-1}=\Omega(s/\sqrt{n})$ support points. Thus each recursive step either shortens the input length or removes a substantial chunk of the support. Choosing $r$ to balance these two effects yields the $\widetilde{O}(n^{3/4}\sqrt{s})$ bound when $b\le n$.

On the lower-bound side, for an $n$-qubit pure state $\ket{\psi}$, the stabilizer nullity is defined as $\nu(\ket{\psi}) := n-\rank(\stab(\ket{\psi}))$, and the exact $T$-count of $\ket{\psi}$ is at least $\nu(\ket{\psi})$~\cite{beverland2020lower}. Using this fact, \cite{vilmart2025resource} showed that exactly preparing $\ket{W_s}\ket{0^{n-s}}$ requires $\Omega(s)$ $T$ gates, since $\nu(\ket{W_s}\ket{0^{n-s}})=s-1$. We extend this linear lower bound to approximate preparation by proving that, for every state $\ket{\phi}$ with $D_{\mathrm{tr}}(\ket{\phi}, \ket{W_s}\ket{0^{n-s}})\le \epsilon$, the stabilizer subgroup satisfies $\rank(\stab(\phi)) \le n-(1-\epsilon^2)s+1+\log\frac{2}{1-\epsilon^2}$; this implies an $\Omega((1-\epsilon^2)s)$ lower bound. In the large-support regime $s=\Omega(n)$, we instead construct a large family of pairwise well-separated $s$-sparse states via a packing lemma, and combine it with the adaptive Clifford+$T$ counting theorem of \cite{gosset2024quantum}, which bounds the number of states preparable with $t$ $T$ gates by $2^{O(n^2+t^2)}$; comparing the two bounds gives an $\Omega(\sqrt{ns})$ lower bound. 

\paragraph{Discussion.}
Our results show that the previous linear $O(s)$ dependence is not inherent once the support is sufficiently large. The main remaining question is to close the gap between the upper bound $\widetilde{O}(n^{3/4}\sqrt{s})$ and the lower bound $\Omega(\sqrt{ns})$. It is also unclear whether the crossover around $s=n^{3/2}$ up to polylogarithmic factors and the factor $n^{3/4}$ are artifacts of the present algorithm or reflect inherent limitations of fault-tolerant sparse quantum state preparation.

Another direction is to exploit more directly the partial nature of the Boolean functions that arise in sparse state preparation. Our support-aware synthesis theorem implements a total Boolean function, whereas the state-preparation application only requires correctness on the support of the target state. The known $O(s)$ upper bound in the regime $s=O(n)$ already benefits from this relaxation. By contrast, when $n$ is sufficiently large, incorporating this relaxation into our construction does not improve the present asymptotic bound, and the construction therefore does not use it. Whether partial-function structure can be exploited in a different way remains an interesting problem.

\paragraph{Paper organization.}
The remainder of this paper is organized as follows. Section~\ref{sec:preliminaries} introduces notation, the synthesis model, and the basic tools that we use. Section~\ref{sec:sqsp} presents the sparse-state preparation algorithm and proves the main upper bound on the $T$-count. Section~\ref{sec:sparse-total-boolean} develops the support-aware Boolean-synthesis framework underlying the construction. Section~\ref{sec:lower-bound} proves the lower bounds for sparse-state preparation. Section~\ref{sec:conclusion-framework} concludes with a summary of the main results and their implications.

\section{Preliminaries}
\label{sec:preliminaries}

Throughout, logarithms are base $2$. For a positive integer $k$ we write
$[k] := \{0,1,\dots,k-1\}$.

\paragraph{Synthesis model.}
An \emph{Clifford+$T$ circuit} is a quantum circuit generated by H, S, CNOT, and T gates.
An \emph{NCT circuit} is a reversible circuit generated by NOT, CNOT, and Toffoli gates. 
Since each Toffoli gate has a constant-size Clifford+$T$ decomposition, any Toffoli-count bound implies the same asymptotic $T$-count bound. We therefore freely translate between Toffoli count and $T$-count for NCT circuits.

\paragraph{Boolean function.}
For a total Boolean function \(F:\{0,1\}^n\to\{0,1\}^b\), write
\begin{align}
\operatorname{supp}(F):=\{x\in\{0,1\}^n : F(x)\neq 0^b\}.
\end{align}

A reversible implementation of a Boolean function $F : \{0,1\}^n \to \{0,1\}^b$ is a unitary $U_F$ such that
\begin{align}
U_F |x\rangle |y\rangle |0^a\rangle
=
|x\rangle |y \oplus F(x)\rangle |0^a\rangle 
\end{align}
for every input $x \in \{0,1\}^n$, output register $y \in \{0,1\}^b$.

The most direct way to implement a Boolean function is to apply unary iteration~\cite{babbush2018encoding, khattar2025rise}, which requires $O(2^n)$ Toffoli gates and $O(n)$ ancillas, regardless of the output length $b$. This Toffoli bound could be improved when $b$ is small.

\begin{lemma}[Boolean-function synthesis {\cite{low2024trading}}]
\label{lem:total-boolean}
For every total Boolean function $F : \{0,1\}^n \to \{0,1\}^b$, there exists an NCT circuit implementing $F$ with
$O\left(\sqrt{b2^n}\right)$
Toffoli gates and ancillas.
\end{lemma}

The sparse-state construction of \cite{rupprecht2026sparse} gives an efficient way to compress a sparse support set into a logarithmic-size label register, which is illustrated in \cref{fig:sparse-indexing-supported}.

\begin{lemma}[Sparse-support indexing {\cite{rupprecht2026sparse}}]
\label{lem:sparse-indexing}
Let $S \subseteq \{0,1\}^n$ with $|S| = s$, and let
$\ell := \lceil \log_2 s \rceil$. Then there exist an injection
$\iota:S\to\{0,1\}^{\ell}$ and an NCT circuit $W_S$ using
$O(s)$ Toffoli gates and $a=O(\log s)$ auxiliary qubits such that,
for every $x\in\{0,1\}^n$,
\[
W_S\ket{x}\ket{0^\ell}\ket{0^a}
=
\ket{u_x}\ket{t_x}\ket{0^a}
\]
for some $u_x\in\{0,1\}^n$ and $t_x\in\{0,1\}^{\ell}$, and, for every
$x\in S$, $(u_x, t_x) = (0^n, \iota(x)).$
\end{lemma}

\begin{figure}[htbp]
\centering
\resizebox{0.35\textwidth}{!}{\begin{tikzpicture}[scale=0.850000,x=1pt,y=1pt]
\filldraw[color=white] (0.000000, -7.000000) rectangle (58.000000, 35.000000);
% Drawing wires
% Line 6: x W \ket{x \in S} \ket{0^n}
\draw[color=black] (0.000000,28.000000) -- (58.000000,28.000000);
\draw[color=black] (0.000000,28.000000) node[left] {$\ket{x \in S}$};
% Line 7: lab W \ket{0^\ell} \ket{\iota(x)}
\draw[color=black] (0.000000,14.000000) -- (58.000000,14.000000);
\draw[color=black] (0.000000,14.000000) node[left] {$\ket{0^\ell}$};
% Line 8: anc W \ket{0^a} \ket{0^a}
\draw[color=black] (0.000000,0.000000) -- (58.000000,0.000000);
\draw[color=black] (0.000000,0.000000) node[left] {$\ket{0^a}$};
% Done with wires; drawing gates
% Line 10: x lab anc G:width=38 $W_S$
\draw (29.000000,28.000000) -- (29.000000,0.000000);
\begin{scope}
\draw[fill=white] (29.000000, 14.000000) +(-45.000000:26.870058pt and 28.284271pt) -- +(45.000000:26.870058pt and 28.284271pt) -- +(135.000000:26.870058pt and 28.284271pt) -- +(225.000000:26.870058pt and 28.284271pt) -- cycle;
\clip (29.000000, 14.000000) +(-45.000000:26.870058pt and 28.284271pt) -- +(45.000000:26.870058pt and 28.284271pt) -- +(135.000000:26.870058pt and 28.284271pt) -- +(225.000000:26.870058pt and 28.284271pt) -- cycle;
\draw (29.000000, 14.000000) node {$W_S$};
\end{scope}
% Done with gates; drawing ending labels
\draw[color=black] (58.000000,28.000000) node[right] {$\ket{0^n}$};
\draw[color=black] (58.000000,14.000000) node[right] {$\ket{\iota(x)}$};
\draw[color=black] (58.000000,0.000000) node[right] {$\ket{0^a}$};
% Done with ending labels; drawing cut lines and comments
% Done with comments
\end{tikzpicture}}
\caption{Action of the sparse-support indexing circuit $W_S$ on supported inputs. For every $x\in S$, the circuit maps $\ket{x}\ket{0^\ell}\ket{0^a}$ to $\ket{0^n}\ket{\iota(x)}\ket{0^a}$, where $\iota:S\to\{0,1\}^{\ell}$ is injective.}
\label{fig:sparse-indexing-supported}
\end{figure}
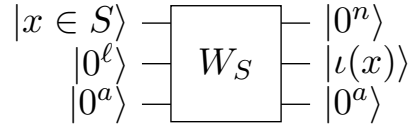

\section{Sparse Quantum State Preparation}
\label{sec:sqsp}
In this section, we derive the sparse state preparation upper bound. We first record the best known bound for preparing a general quantum state below.
\begin{lemma}[General state preparation {\cite{gosset2024quantum}}]
\label{lem:general-state-preparation}
Any $n$-qubit state can be prepared up to error $\epsilon$ by a Clifford+$T$ circuit starting with the all-zeros state using $O(\sqrt{2^n\log(1/\epsilon)} + \log(1/\epsilon))$
$T$ gates and ancillas.
\end{lemma}

Another tool we need is the following sparse support Boolean function synthesis theorem. We defer its proof to \cref{sec:sparse-total-boolean}.
 
\begin{theorem}
\label{thm:sparse-total-boolean}
Let $F:\{0,1\}^n\to\{0,1\}^b$ be a Boolean function with $1 \le |\operatorname{supp}(F)| = s$.
If $s \le n^{3/2}\log n$, then $F$ can be implemented by an NCT circuit with $O(s+n)$ Toffoli gates and $O(\log s)$ ancillas.
If $s>n^{3/2}\log n$, then $F$ can be implemented by an NCT circuit with
$O(n^{3/4}\sqrt{s\log\frac{s}{\sqrt{n}}}+n^{1/4}\sqrt{bs})$
Toffoli gates and
$O(n\log\frac{s}{\sqrt{n}}+n^{1/4}\sqrt{s}+n^{-1/4}\sqrt{bs})$
ancillas.
\end{theorem}

We now state the sparse-state preparation theorem.
\begin{theorem}
\label{thm:sqsp-main}
For every $n$-qubit $s$-sparse state and every $\epsilon>0$, there exists a Clifford+$T$ circuit starting with the all-zeros state that prepares it up to error $\epsilon$ with the following resource bounds.
\begin{itemize}
    \item If $s\le n^{3/2}\log n$, the circuit uses $O(s+\log(1/\epsilon))$ $T$ gates and $O(\log s+\log(1/\epsilon))$ ancillas.
    \item If $s>n^{3/2}\log n$, the circuit uses $O(n^{3/4}\sqrt{s\log\frac{s}{\sqrt{n}}}+\sqrt{s\log(1/\epsilon)}+\log(1/\epsilon))$ $T$ gates and $O(n\log\frac{s}{\sqrt{n}}+\sqrt{ns}+\sqrt{s\log(1/\epsilon)}+\log(1/\epsilon))$ ancillas.
\end{itemize}

\end{theorem}

\begin{proof}
If $s \le n^{3/2}\log n$, then the claim follows from the sparse-state preparation construction of \cite{rupprecht2026sparse}, together with the observation that the dense-state preparation step can be implemented using $O(s)$ $T$ gates and $O(\log(1/\epsilon))$ ancillas~\cite{gosset2024quantum}.

It remains to prove the case where $s > n^{3/2}\log n$. We count the reversible subroutines below in Toffoli gates and convert them to $T$-count at the end. The construction has three steps: prepare a label state, coherently load the support strings, and then erase the label register via the inverse labeling map.

Write the target state as
\begin{align}
|\psi\rangle = \sum_{j=0}^{s-1} \alpha_j |x^{(j)}\rangle,
\end{align}
where $x^{(0)},\ldots,x^{(s-1)} \in \{0,1\}^n$ are distinct basis strings, and let
$\ell := \lceil \log_2 s \rceil$.

We first prepare the $\ell$-qubit label state
\begin{align}
|\phi\rangle := \sum_{j=0}^{s-1} \alpha_j |j\rangle.
\end{align}
By \cref{lem:general-state-preparation}, there is a Clifford+$T$ circuit that prepares an $\epsilon$-approximation of $|\phi\rangle$ from $|0^\ell\rangle$ using
\begin{align}
O\!\left(\sqrt{2^\ell\log(1/\epsilon)} + \log(1/\epsilon)\right)
= O\!\left(\sqrt{s\log(1/\epsilon)} + \log(1/\epsilon)\right)
\end{align}
$T$ gates and ancillas.

Next we load the support strings into an $n$-qubit data register. Define a Boolean function
\begin{align}
F : \{0,1\}^{\ell} \to \{0,1\}^{n}
\end{align}
by
\begin{align}
F(j) &= x^{(j)} \qquad \text{for every } j=0,\ldots,s-1, \\
F(y) &= 0^n \qquad \text{for every } y\notin \{0,\ldots,s-1\}.
\end{align}
Then $|\operatorname{supp}(F)| \le s$. By \cref{lem:total-boolean}, $F$ has an exact NCT implementation using
$O(\sqrt{n\,2^\ell})= O(\sqrt{ns})$
Toffoli gates and $O\!\left(\sqrt{ns}\right)$ ancillas, again since $s\le 2^\ell < 2s$. Applied to the label register and a data register initialized to $|0^n\rangle$, this circuit maps
\begin{align}
\sum_{j=0}^{s-1} \alpha_j |j\rangle |0^n\rangle
\longmapsto
\sum_{j=0}^{s-1} \alpha_j |j\rangle |x^{(j)}\rangle.
\end{align}

It remains to erase the label register. Define the zero-extension of the inverse labeling map by
\begin{align}
L : \{0,1\}^n \to \{0,1\}^{\ell},
\qquad
L(x)
:=
\begin{cases}
 j, & x=x^{(j)} \text{ for some } j=0,\ldots,s-1,\\
0^\ell, & \text{otherwise}.
\end{cases}
\end{align}
Again $|\operatorname{supp}(L)| \le s$. \cref{thm:sparse-total-boolean} gives an exact NCT implementation of $L$ using
\begin{align}
O\!\left(
n^{3/4}\sqrt{s\log\frac{s}{\sqrt{n}}}
+
n^{1/4}\sqrt{\ell s}
\right) = 
O\!\left(
n^{3/4}\sqrt{s\log\frac{s}{\sqrt{n}}}
\right)
\end{align}
Toffoli gates and
\begin{align}
O\!\left(
n\log\frac{s}{\sqrt{n}}
+
n^{1/4}\sqrt{s} + n^{-1/4}\sqrt{\ell s}
\right) =
O\!\left(
n\log\frac{s}{\sqrt{n}}
+
n^{1/4}\sqrt{s}
\right)
\end{align}
ancillas. 
Therefore, applying $L$ with the data register as input and the label register as target yields
\begin{align}
\sum_{j=0}^{s-1} \alpha_j |j\rangle |x^{(j)}\rangle
\longmapsto
\sum_{j=0}^{s-1} \alpha_j |0^\ell\rangle |x^{(j)}\rangle.
\end{align}
Because $F$ and $L$ are implemented exactly, applying them preserves the trace-distance error.

The Toffoli count satisfies 
\begin{align}
&O\!\left(
n^{3/4}\sqrt{s\log\frac{s}{\sqrt{n}}}
 + \sqrt{ns} + \sqrt{s\log(1/\epsilon)} + \log(1/\epsilon)
\right) \\
=\; & O\!\left(n^{3/4}\sqrt{s\log\frac{s}{\sqrt{n}}} + \sqrt{s\log(1/\epsilon)} + \log(1/\epsilon)\right).
\end{align}
The ancilla count satisfies 
\begin{align}
&O\!\left(n\log\frac{s}{\sqrt{n}}+n^{1/4}\sqrt{s} + \sqrt{ns}+\sqrt{s\log(1/\epsilon)} + \log(1/\epsilon)\right) \\
=\; &O\!\left(n\log\frac{s}{\sqrt{n}} + \sqrt{ns}+\sqrt{s\log(1/\epsilon)} + \log(1/\epsilon)\right).
\end{align}
Translating the Toffoli count to $T$-count gives the claimed bound.
\end{proof}

% \begin{corollary}[Combined upper bound]
% \label{cor:sqsp-combined}
% For every $n$-qubit $s$-sparse state and every $\epsilon>0$, there exists a unitary Clifford+$T$ circuit starting with the all-zeros state that prepares it up to error $\epsilon$ using
% \begin{align}
% \widetilde{O}\!\left(
% \min\{s,\ n^{3/4}\sqrt{s}\}
% +
% \sqrt{s\log(1/\epsilon)}
% +
% \log(1/\epsilon)
% \right)
% \end{align}
% $T$ gates.
% \end{corollary}

\section{Support-Aware Synthesis of Sparse Boolean Functions}
\label{sec:sparse-total-boolean}
% In this section, we prove support-aware synthesis bounds for Boolean functions with sparse nonzero support, i.e., \cref{thm:sparse-total-boolean}.
% The small-support and large-support regimes require different ideas.
% In the small-support regime, we compress the supported inputs into logarithmic-size labels and reduce the problem to the task of implementing another function on a logarithmic-size register.
% In the large-support regime, we fix a prefix length $r$ and recursively inspect the active $r$-bit prefixes of the current support: if there are few such prefixes, we compress them and shorten the input length by one; if there are many, we peel one support point from each active prefix and thereby remove a large chunk of the support.

In this section, we prove \cref{thm:sparse-total-boolean}, stated earlier in \cref{sec:sqsp}.
The argument splits according to the support size.
In the small-support regime, we compress the supported inputs into logarithmic-size labels, reducing the problem to Boolean synthesis on an $O(\log s)$-bit register.
In the large-support regime, we fix a prefix length $r$ and recursively inspect the active $r$-bit prefixes of the current support.
When only a few prefixes are active, we injectively relabel them and shorten the input length by one; when many prefixes are active, we peel one support point from each active prefix and thereby remove a large part of the support.
Thus every recursive step replaces the current instance by a simpler one, either with shorter input length or with smaller support.
We establish the two regimes in \cref{lem:sparse-supp-boolean,lem:sparse-total-boolean-large}; together they imply \cref{thm:sparse-total-boolean}.

\begin{proof}[Proof of \cref{thm:sparse-total-boolean}]
If $s\le n^{3/2}\log n$, the claimed bound is exactly \cref{lem:sparse-supp-boolean}.
If $s>n^{3/2}\log n$, the claimed bound is exactly \cref{lem:sparse-total-boolean-large}.
Combining these two lemmas yields the theorem.
\end{proof}

\subsection{Small-support regime}
When the support size is small, we can compress the supported inputs into short labels and reduce the problem to a Boolean function on only $O(\log s)$ input bits.

\begin{lemma}
\label{lem:sparse-supp-boolean}
For every Boolean function $F : \{0,1\}^n \to \{0,1\}^b$ with $1 \le |\operatorname{supp}(F)| = s$, there exists an NCT circuit implementing $F$ with $O(s + n)$ Toffoli gates and $O(\log s)$ ancillas.
\end{lemma}
\begin{proof}
Let $S:=\operatorname{supp}(F)$, and let $\ell:=\lceil \log s\rceil$.
The construction has two stages.
First, \cref{lem:sparse-indexing} assigns each supported input a unique label while mapping the original $n$-bit input register to $0^n$.
Second, a single flag bit records whether this transformed input register is all zero, so the remaining task is just a Boolean function on the label together with the flag.

By \cref{lem:sparse-indexing}, there is an injection $\iota:S\to\{0,1\}^{\ell}$ and an NCT circuit $W_S$ using $O(s)$ Toffoli gates and $a=O(\log s)$ ancillas such that
\begin{align}
W_S|x\rangle |0\rangle^{\otimes \ell}|0\rangle^{\otimes a}
=
\ket{0^n}|\iota(x)\rangle |0\rangle^{\otimes a}
\end{align}
for every $x\in S$.

Define a Boolean function $H:\{0,1\}^{\ell+1}\to\{0,1\}^b$
by
\begin{align}
H(z,t):=
\begin{cases}
F(x), & z=1 \text{ and } t=\iota(x)\text{ for some }x\in S,\\
0^b, & \text{otherwise}.
\end{cases}
\end{align}
This is well defined because $\iota$ is injective.
By unary iteration~\cite{babbush2018encoding, khattar2025rise}, there exists an NCT circuit $U_H$ implementing $H$ with $O(2^{\ell+1}) = O(s)$ Toffoli gates and $O(\log s)$ ancillas.

We implement $F$ as follows, which is illustrated in \cref{fig:sparse-small-support-boolean}. Suppressing the ancillas used internally by $W_S$ and $U_H$, we start from $|x\rangle |y\rangle |0\rangle^{\otimes \ell}|0\rangle$, where the registers store, respectively, the input, the output, the label, and the flag bit.
We first apply $W_S$ to the input register and the label register. Next we compute a flag bit $z$ indicating whether the input register equals $0^{n}$. This flag can be computed by an $n$-controlled Toffoli gate with open controls on the input bits; together with its uncomputation, it requires $O(n)$ Toffoli gates and $O(1)$ ancillas~\cite{khattar2025rise}. We then apply $U_H$ to the flag and label registers and the output register, and finally uncompute the flag and apply $W_S^\dagger$.

If $x\in S$, then after applying $W_S$ and computing this flag we obtain $|0\rangle^{\otimes n}|y\rangle |\iota(x)\rangle |1\rangle$ on the input, output, label, and flag registers. The circuit $U_H$ therefore maps the output register to $|y\oplus H(1,\iota(x))\rangle=|y\oplus F(x)\rangle$. Uncomputing the flag and $W_S$ returns all ancillas to $|0\rangle$, so the overall action is $|x\rangle |y\rangle \longmapsto |x\rangle |y\oplus F(x)\rangle$.

Now suppose $x\notin S$, and, after suppressing the clean auxiliary bits of $W_S$, write
\begin{align}
W_S|x\rangle |0\rangle^{\otimes \ell}
= |u\rangle |t\rangle .
\end{align}
If the flag is $0$, then $H$ outputs $0^b$ by definition, so the output register is unchanged. If the flag is $1$, then $u=0^n$. We claim that $t\notin \iota(S)$. Indeed, if $t=\iota(x')$ for some $x'\in S$, then
\begin{align}
W_S|x\rangle |0\rangle^{\otimes \ell}
=
\ket{0^n}|\iota(x')\rangle
=
W_S|x'\rangle |0\rangle^{\otimes \ell},
\end{align}
contradicting the injectivity of the action of $W_S$ on computational-basis states. Hence $H(1,t)=0^b$, and the output register again remains unchanged. Therefore the circuit computes $F$ exactly on all inputs.

The Toffoli cost is the sum of the costs of $W_S$ and $W_S^\dagger$, the flag computation and its uncomputation, and $U_H$, namely, $O(s + n)$.
The ancilla count is $O(\log s)$. This proves the lemma.
\end{proof}

\begin{figure}[htbp]
\centering
\resizebox{0.8\textwidth}{!}{\begin{tikzpicture}[scale=0.850000,x=1pt,y=1pt]
\filldraw[color=white] (0.000000, -7.000000) rectangle (232.000000, 49.000000);
% Drawing wires
% Line 6: x W \ket{x} \ket{x}
\draw[color=black] (0.000000,42.000000) -- (232.000000,42.000000);
\draw[color=black] (0.000000,42.000000) node[left] {$\ket{x}$};
% Line 7: lab W \ket{0^\ell} \ket{0^\ell}
\draw[color=black] (0.000000,28.000000) -- (232.000000,28.000000);
\draw[color=black] (0.000000,28.000000) node[left] {$\ket{0^\ell}$};
% Line 8: z W \ket{0} \ket{0}
\draw[color=black] (0.000000,14.000000) -- (232.000000,14.000000);
\draw[color=black] (0.000000,14.000000) node[left] {$\ket{0}$};
% Line 9: y W \ket{y} \ket{y\oplus F(x)}
\draw[color=black] (0.000000,0.000000) -- (232.000000,0.000000);
\draw[color=black] (0.000000,0.000000) node[left] {$\ket{y}$};
% Done with wires; drawing gates
% Line 11: x lab G:width=42 $W_S$
\draw (31.000000,42.000000) -- (31.000000,28.000000);
\begin{scope}
\draw[fill=white] (31.000000, 35.000000) +(-45.000000:29.698485pt and 18.384776pt) -- +(45.000000:29.698485pt and 18.384776pt) -- +(135.000000:29.698485pt and 18.384776pt) -- +(225.000000:29.698485pt and 18.384776pt) -- cycle;
\clip (31.000000, 35.000000) +(-45.000000:29.698485pt and 18.384776pt) -- +(45.000000:29.698485pt and 18.384776pt) -- +(135.000000:29.698485pt and 18.384776pt) -- +(225.000000:29.698485pt and 18.384776pt) -- cycle;
\draw (31.000000, 35.000000) node {$W_S$};
\end{scope}
% Line 12: TOUCH
% Line 13: +z -x
\draw (75.000000,42.000000) -- (75.000000,14.000000);
\begin{scope}
\draw[fill=white] (75.000000, 14.000000) circle(3.000000pt);
\clip (75.000000, 14.000000) circle(3.000000pt);
\draw (72.000000, 14.000000) -- (78.000000, 14.000000);
\draw (75.000000, 11.000000) -- (75.000000, 17.000000);
\end{scope}
\draw[fill=white] (75.000000, 42.000000) circle(2.250000pt);
% Line 14: TOUCH
% Line 15: lab z y G:width=36 $U_H$
\draw (116.000000,28.000000) -- (116.000000,0.000000);
\begin{scope}
\draw[fill=white] (116.000000, 14.000000) +(-45.000000:25.455844pt and 28.284271pt) -- +(45.000000:25.455844pt and 28.284271pt) -- +(135.000000:25.455844pt and 28.284271pt) -- +(225.000000:25.455844pt and 28.284271pt) -- cycle;
\clip (116.000000, 14.000000) +(-45.000000:25.455844pt and 28.284271pt) -- +(45.000000:25.455844pt and 28.284271pt) -- +(135.000000:25.455844pt and 28.284271pt) -- +(225.000000:25.455844pt and 28.284271pt) -- cycle;
\draw (116.000000, 14.000000) node {$U_H$};
\end{scope}
% Line 16: TOUCH
% Line 17: +z -x
\draw (157.000000,42.000000) -- (157.000000,14.000000);
\begin{scope}
\draw[fill=white] (157.000000, 14.000000) circle(3.000000pt);
\clip (157.000000, 14.000000) circle(3.000000pt);
\draw (154.000000, 14.000000) -- (160.000000, 14.000000);
\draw (157.000000, 11.000000) -- (157.000000, 17.000000);
\end{scope}
\draw[fill=white] (157.000000, 42.000000) circle(2.250000pt);
% Line 18: TOUCH
% Line 19: x lab G:width=42 $W_S^\dagger$
\draw (201.000000,42.000000) -- (201.000000,28.000000);
\begin{scope}
\draw[fill=white] (201.000000, 35.000000) +(-45.000000:29.698485pt and 18.384776pt) -- +(45.000000:29.698485pt and 18.384776pt) -- +(135.000000:29.698485pt and 18.384776pt) -- +(225.000000:29.698485pt and 18.384776pt) -- cycle;
\clip (201.000000, 35.000000) +(-45.000000:29.698485pt and 18.384776pt) -- +(45.000000:29.698485pt and 18.384776pt) -- +(135.000000:29.698485pt and 18.384776pt) -- +(225.000000:29.698485pt and 18.384776pt) -- cycle;
\draw (201.000000, 35.000000) node {$W_S^\dagger$};
\end{scope}
% Done with gates; drawing ending labels
\draw[color=black] (232.000000,42.000000) node[right] {$\ket{x}$};
\draw[color=black] (232.000000,28.000000) node[right] {$\ket{0^\ell}$};
\draw[color=black] (232.000000,14.000000) node[right] {$\ket{0}$};
\draw[color=black] (232.000000,0.000000) node[right] {$\ket{y\oplus F(x)}$};
% Done with ending labels; drawing cut lines and comments
% Done with comments
\end{tikzpicture}}
\caption{High-level circuit for the small-support implementation of $F$. The open-control CNOTs compute and uncompute a flag $z$ indicating whether the input register equals $0^n$. The box $U_H$ then uses the flag and the label register to update the output register.}
\label{fig:sparse-small-support-boolean}
\end{figure}
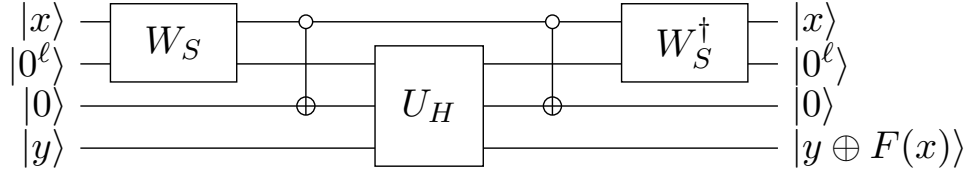

\subsection{Large-support regime}
When the support is larger, the direct sparse-support construction is no longer competitive, so we switch to a recursive argument.

\begin{lemma}
\label{lem:sparse-total-boolean-large}
Let $F:\{0,1\}^n\to\{0,1\}^b$ be a Boolean function with $|\operatorname{supp}(F)| = s$, and assume that $s>n^{3/2}\log n$.
Then $F$ can be implemented by an NCT circuit with
$O(n^{3/4}\sqrt{s\log\frac{s}{\sqrt{n}}}+n^{1/4}\sqrt{bs})$
Toffoli gates and $O(n\log\frac{s}{\sqrt{n}}+n^{1/4}\sqrt{s}+n^{-1/4}\sqrt{bs})$ ancillas.
\end{lemma}

\begin{proof}
Let $c$ be a parameter satisfying $1\le c < s$, to be chosen later, and set $r:=\left\lceil \log(s/c)\right\rceil$.
We fix this value of $r$ for the entire recursion.
At a generic recursive call, let $G:\{0,1\}^m\to\{0,1\}^b$ be the current Boolean function, let $S_G:=\operatorname{supp}(G)$, and let $t:=|S_G|$.
Initially, $(G,m,t)=(F,n,s)$.

If $t\le n^{3/2}\log n$, we stop the recursion and apply \cref{lem:sparse-supp-boolean} to $G$. Since $m\le n$, this leaf costs $O(t+m)=O(n^{3/2}\log n)$ Toffoli gates and $O(\log n)$ ancillas.

If $t>n^{3/2}\log n$ and $m=r$, we stop the recursion and apply \cref{lem:total-boolean}. Since $2^r\le 2s/c$, this leaf costs $O(\sqrt{b\,2^r})=O(\sqrt{\frac{bs}{c}})$ Toffoli gates and ancillas.

Assume from now on that $t>n^{3/2}\log n$ and $m>r$. Write each input as $x=uv$, where $u\in\{0,1\}^r$ and $v\in\{0,1\}^{m-r}$, and let $U_G:=\{u\in\{0,1\}^r:\exists v,\ uv\in S_G\}$ and $p_G:=|U_G|$.
The set $U_G$ records which $r$-bit prefixes are active on the current support. The recursive step now chooses between prefix compression and support peeling according to the size of $U_G$.

\paragraph{Case 1: few active prefixes.}
Assume that $p_G\le 2^{r-1}-1$.
This is exactly the regime in which the active prefixes fit into the nonzero strings of length $r-1$.
We therefore reserve $0^{r-1}$ as a default code for inactive prefixes and use the remaining labels to relabel the active prefixes injectively.

Choose an injection $\gamma_G:U_G\to \{0,1\}^{r-1}\setminus\{0^{r-1}\}$, and extend it to a total map $\bar{\gamma}_G:\{0,1\}^r\to\{0,1\}^{r-1}$ by setting $\bar{\gamma}_G(u)=0^{r-1}$ for $u\notin U_G$. The point of excluding $0^{r-1}$ from the image of $\gamma_G$ is that the all-zero label becomes a safe default for inactive prefixes.

Define the compressed function $\widetilde{G}:\{0,1\}^{m-1}\to\{0,1\}^b$ by
\begin{align}
\widetilde{G}(yv)
:=
\begin{cases}
G(uv), & \text{if } y=\gamma_G(u)\text{ and }uv\in S_G,\\
0^b, & \text{otherwise}.
\end{cases}
\end{align}
Because the map $uv\mapsto \gamma_G(u)v$ is injective on $S_G$ (if $\gamma_G(u)v=\gamma_G(u')v'$, then $v=v'$, and injectivity of $\gamma_G$ on $U_G$ gives $u=u'$), the definition of $\widetilde{G}$ is unambiguous. Consequently, $\operatorname{supp}(\widetilde{G})=\{\gamma_G(u)v:uv\in S_G\}$, and hence $\widetilde{G}$ has support size exactly $t$.

To implement $G$, we compute $\bar{\gamma}_G(u)$ into a fresh $(r-1)$-bit register, apply the recursive circuit for $\widetilde{G}$ to the logical input $\bar{\gamma}_G(u)v$, and then uncompute $\bar{\gamma}_G(u)$.
This prefix-compression step is illustrated in \cref{fig:sparse-large-support-case1}.

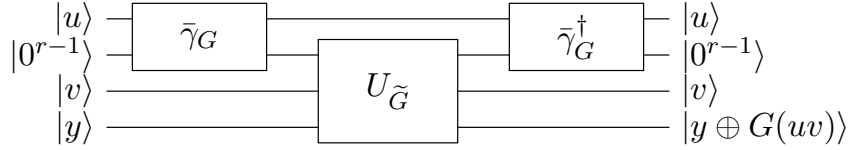
\begin{figure}[htbp]
\centering
\resizebox{0.7\textwidth}{!}{\begin{tikzpicture}[scale=0.850000,x=1pt,y=1pt]
\filldraw[color=white] (0.000000, -7.000000) rectangle (218.000000, 49.000000);
% Drawing wires
% Line 6: u W \ket{u} \ket{u}
\draw[color=black] (0.000000,42.000000) -- (218.000000,42.000000);
\draw[color=black] (0.000000,42.000000) node[left] {$\ket{u}$};
% Line 7: g W \ket{0^{r-1}} \ket{0^{r-1}}
\draw[color=black] (0.000000,28.000000) -- (218.000000,28.000000);
\draw[color=black] (0.000000,28.000000) node[left] {$\ket{0^{r-1}}$};
% Line 8: v W \ket{v} \ket{v}
\draw[color=black] (0.000000,14.000000) -- (218.000000,14.000000);
\draw[color=black] (0.000000,14.000000) node[left] {$\ket{v}$};
% Line 9: y W \ket{y} \ket{y\oplus G(uv)}
\draw[color=black] (0.000000,0.000000) -- (218.000000,0.000000);
\draw[color=black] (0.000000,0.000000) node[left] {$\ket{y}$};
% Done with wires; drawing gates
% Line 11: u g G:width=52 $\bar{\gamma}_G$
\draw (36.000000,42.000000) -- (36.000000,28.000000);
\begin{scope}
\draw[fill=white] (36.000000, 35.000000) +(-45.000000:36.769553pt and 18.384776pt) -- +(45.000000:36.769553pt and 18.384776pt) -- +(135.000000:36.769553pt and 18.384776pt) -- +(225.000000:36.769553pt and 18.384776pt) -- cycle;
\clip (36.000000, 35.000000) +(-45.000000:36.769553pt and 18.384776pt) -- +(45.000000:36.769553pt and 18.384776pt) -- +(135.000000:36.769553pt and 18.384776pt) -- +(225.000000:36.769553pt and 18.384776pt) -- cycle;
\draw (36.000000, 35.000000) node {$\bar{\gamma}_G$};
\end{scope}
% Line 12: TOUCH
% Line 13: g v y G:width=54 $U_{\widetilde{G}}$
\draw (109.000000,28.000000) -- (109.000000,0.000000);
\begin{scope}
\draw[fill=white] (109.000000, 14.000000) +(-45.000000:38.183766pt and 28.284271pt) -- +(45.000000:38.183766pt and 28.284271pt) -- +(135.000000:38.183766pt and 28.284271pt) -- +(225.000000:38.183766pt and 28.284271pt) -- cycle;
\clip (109.000000, 14.000000) +(-45.000000:38.183766pt and 28.284271pt) -- +(45.000000:38.183766pt and 28.284271pt) -- +(135.000000:38.183766pt and 28.284271pt) -- +(225.000000:38.183766pt and 28.284271pt) -- cycle;
\draw (109.000000, 14.000000) node {$U_{\widetilde{G}}$};
\end{scope}
% Line 14: TOUCH
% Line 15: u g G:width=52 $\bar{\gamma}_G^\dagger$
\draw (182.000000,42.000000) -- (182.000000,28.000000);
\begin{scope}
\draw[fill=white] (182.000000, 35.000000) +(-45.000000:36.769553pt and 18.384776pt) -- +(45.000000:36.769553pt and 18.384776pt) -- +(135.000000:36.769553pt and 18.384776pt) -- +(225.000000:36.769553pt and 18.384776pt) -- cycle;
\clip (182.000000, 35.000000) +(-45.000000:36.769553pt and 18.384776pt) -- +(45.000000:36.769553pt and 18.384776pt) -- +(135.000000:36.769553pt and 18.384776pt) -- +(225.000000:36.769553pt and 18.384776pt) -- cycle;
\draw (182.000000, 35.000000) node {$\bar{\gamma}_G^\dagger$};
\end{scope}
% Done with gates; drawing ending labels
\draw[color=black] (218.000000,42.000000) node[right] {$\ket{u}$};
\draw[color=black] (218.000000,28.000000) node[right] {$\ket{0^{r-1}}$};
\draw[color=black] (218.000000,14.000000) node[right] {$\ket{v}$};
\draw[color=black] (218.000000,0.000000) node[right] {$\ket{y\oplus G(uv)}$};
% Done with ending labels; drawing cut lines and comments
% Done with comments
\end{tikzpicture}}
\caption{Prefix-compression step in Case 1 of the large-support recursion. The circuit computes $\bar{\gamma}_G(u)$ into a fresh $(r-1)$-bit register, applies the recursive implementation $U_{\widetilde{G}}$ to the logical input $\bar{\gamma}_G(u)v$, and then uncomputes $\bar{\gamma}_G(u)$. This reduces the input length from $m$ to $m-1$ while preserving the support size.}
\label{fig:sparse-large-support-case1}
\end{figure}

This works because:
\begin{itemize}
\item if $uv\in S_G$, then $\bar{\gamma}_G(u)=\gamma_G(u)$ and $\widetilde{G}(\bar{\gamma}_G(u)v)=G(uv)$;
\item if $uv\notin S_G$ but $u\in U_G$, then $\widetilde{G}(\gamma_G(u),v)=0^b$: otherwise there would exist $u'\in U_G$ with $\gamma_G(u')=\gamma_G(u)$ and $u'v\in S_G$, and injectivity of $\gamma_G$ would force $u'=u$, contradicting $uv\notin S_G$;
\item if $u\notin U_G$, then $\bar{\gamma}_G(u)=0^{r-1}$, and this compressed prefix never appears on the support of $\widetilde{G}$.
\end{itemize}

Thus Case 1 strictly decreases the input length, while keeping the support size unchanged.
The only non-recursive cost is the synthesis of $\bar{\gamma}_G$, which by \cref{lem:total-boolean} costs $O(\sqrt{(r-1)\,2^r})=O(\sqrt{\frac{s}{c}\log\frac{s}{c}})$ Toffoli gates and ancillas.

\paragraph{Case 2: many active prefixes.}
Assume now that $p_G\ge 2^{r-1}$.
In this regime there are many active prefixes.
Since every $u\in U_G$ supports at least one point of $S_G$, we can choose one representative from each active prefix and remove all of them in one shot.
For each $u\in U_G$, choose one representative suffix $\tau_G(u)\in\{0,1\}^{m-r}$ such that $u\tau_G(u)\in S_G$. Define $A_G:=\{u\tau_G(u):u\in U_G\}$ and $S_{\mathrm{rem}}:=S_G\setminus A_G$.
Thus $A_G$ contains exactly one support point for each active prefix, so $|A_G|=p_G$.
Define
\begin{equation}
\begin{aligned}
G_{\mathrm{peel}}(x)&:=
\begin{cases}
G(x), & x\in A_G,\\
0^b, & x\notin A_G,
\end{cases}
\\
G_{\mathrm{rem}}(x)&:=
\begin{cases}
G(x), & x\in S_{\mathrm{rem}},\\
0^b, & x\notin S_{\mathrm{rem}}.
\end{cases}
\end{aligned}
\end{equation}
Then $G=G_{\mathrm{rem}}\oplus G_{\mathrm{peel}}$ because $A_G\cap S_{\mathrm{rem}}=\varnothing$, and $|\operatorname{supp}(G_{\mathrm{rem}})|=t-p_G$.

To implement $G_{\mathrm{peel}}$, it suffices to detect whether the suffix equals the chosen representative for the current prefix and, if so, output the corresponding value.
For this purpose, extend $\tau_G$ to a total map $\bar{\tau}_G:\{0,1\}^r\to\{0,1\}^{m-r}$ by setting $\bar{\tau}_G(u)=0^{m-r}$ for $u\notin U_G$.
Also define $H_G:\{0,1\}^{r+1}\to\{0,1\}^b$ by
\begin{align}
H_G(z,u):=
\begin{cases}
G_{\mathrm{peel}}(u\tau_G(u)), & z=1 \text{ and } u\in U_G,\\
0^b, & \text{otherwise}.
\end{cases}
\end{align}

We implement $G_{\mathrm{peel}}$ as follows on input $x=uv$. First XOR $\bar{\tau}_G(u)$ into the suffix register. Next compute a flag bit $z$ indicating whether the modified suffix equals $0^{m-r}$. Then apply the circuit for $H_G$ to the pair $(z,u)$ and the output register. Finally uncompute $z$ and undo the XOR by $\bar{\tau}_G(u)$.
The peeling step, followed by the recursive implementation of $G_{\mathrm{rem}}$, is illustrated in \cref{fig:sparse-large-support-case2}.

\begin{figure}[htbp]
\centering
\resizebox{0.85\textwidth}{!}{\begin{tikzpicture}[scale=0.850000,x=1pt,y=1pt]
\filldraw[color=white] (0.000000, -9.500000) rectangle (272.000000, 66.500000);
% Drawing wires
% Line 6: u W \ket{u} \ket{u}
\draw[color=black] (0.000000,57.000000) -- (272.000000,57.000000);
\draw[color=black] (0.000000,57.000000) node[left] {$\ket{u}$};
% Line 7: v W \ket{v} \ket{v}
\draw[color=black] (0.000000,38.000000) -- (272.000000,38.000000);
\draw[color=black] (0.000000,38.000000) node[left] {$\ket{v}$};
% Line 8: z W \ket{0} \ket{0}
\draw[color=black] (0.000000,19.000000) -- (272.000000,19.000000);
\draw[color=black] (0.000000,19.000000) node[left] {$\ket{0}$};
% Line 9: y W \ket{y} \ket{y\oplus G(uv)}
\draw[color=black] (0.000000,0.000000) -- (272.000000,0.000000);
\draw[color=black] (0.000000,0.000000) node[left] {$\ket{y}$};
% Done with wires; drawing gates
% Line 11: u v G:width=35 $\bar{\tau}_G$
\draw (27.500000,57.000000) -- (27.500000,38.000000);
\begin{scope}
\draw[fill=white] (27.500000, 47.500000) +(-45.000000:24.748737pt and 24.748737pt) -- +(45.000000:24.748737pt and 24.748737pt) -- +(135.000000:24.748737pt and 24.748737pt) -- +(225.000000:24.748737pt and 24.748737pt) -- cycle;
\clip (27.500000, 47.500000) +(-45.000000:24.748737pt and 24.748737pt) -- +(45.000000:24.748737pt and 24.748737pt) -- +(135.000000:24.748737pt and 24.748737pt) -- +(225.000000:24.748737pt and 24.748737pt) -- cycle;
\draw (27.500000, 47.500000) node {$\bar{\tau}_G$};
\end{scope}
% Line 12: TOUCH
% Line 13: +z -v
\draw (68.000000,38.000000) -- (68.000000,19.000000);
\begin{scope}
\draw[fill=white] (68.000000, 19.000000) circle(3.000000pt);
\clip (68.000000, 19.000000) circle(3.000000pt);
\draw (65.000000, 19.000000) -- (71.000000, 19.000000);
\draw (68.000000, 16.000000) -- (68.000000, 22.000000);
\end{scope}
\draw[fill=white] (68.000000, 38.000000) circle(2.250000pt);
% Line 14: TOUCH
% Line 15: u G:width=35 $H_G$  z y G:width=35 $H_G$
\draw (108.500000,57.000000) -- (108.500000,0.000000);
\begin{scope}
\draw[fill=white] (108.500000, 57.000000) +(-45.000000:24.748737pt and 11.313708pt) -- +(45.000000:24.748737pt and 11.313708pt) -- +(135.000000:24.748737pt and 11.313708pt) -- +(225.000000:24.748737pt and 11.313708pt) -- cycle;
\clip (108.500000, 57.000000) +(-45.000000:24.748737pt and 11.313708pt) -- +(45.000000:24.748737pt and 11.313708pt) -- +(135.000000:24.748737pt and 11.313708pt) -- +(225.000000:24.748737pt and 11.313708pt) -- cycle;
\draw (108.500000, 57.000000) node {$H_G$};
\end{scope}
\begin{scope}
\draw[fill=white] (108.500000, 9.500000) +(-45.000000:24.748737pt and 24.748737pt) -- +(45.000000:24.748737pt and 24.748737pt) -- +(135.000000:24.748737pt and 24.748737pt) -- +(225.000000:24.748737pt and 24.748737pt) -- cycle;
\clip (108.500000, 9.500000) +(-45.000000:24.748737pt and 24.748737pt) -- +(45.000000:24.748737pt and 24.748737pt) -- +(135.000000:24.748737pt and 24.748737pt) -- +(225.000000:24.748737pt and 24.748737pt) -- cycle;
\draw (108.500000, 9.500000) node {$H_G$};
\end{scope}
% Line 16: TOUCH
% Line 17: +z -v
\draw (149.000000,38.000000) -- (149.000000,19.000000);
\begin{scope}
\draw[fill=white] (149.000000, 19.000000) circle(3.000000pt);
\clip (149.000000, 19.000000) circle(3.000000pt);
\draw (146.000000, 19.000000) -- (152.000000, 19.000000);
\draw (149.000000, 16.000000) -- (149.000000, 22.000000);
\end{scope}
\draw[fill=white] (149.000000, 38.000000) circle(2.250000pt);
% Line 18: TOUCH
% Line 19: u v G:width=35 $\bar{\tau}_G$
\draw (189.500000,57.000000) -- (189.500000,38.000000);
\begin{scope}
\draw[fill=white] (189.500000, 47.500000) +(-45.000000:24.748737pt and 24.748737pt) -- +(45.000000:24.748737pt and 24.748737pt) -- +(135.000000:24.748737pt and 24.748737pt) -- +(225.000000:24.748737pt and 24.748737pt) -- cycle;
\clip (189.500000, 47.500000) +(-45.000000:24.748737pt and 24.748737pt) -- +(45.000000:24.748737pt and 24.748737pt) -- +(135.000000:24.748737pt and 24.748737pt) -- +(225.000000:24.748737pt and 24.748737pt) -- cycle;
\draw (189.500000, 47.500000) node {$\bar{\tau}_G$};
\end{scope}
% Line 20: TOUCH
% Line 21: u v G:width=35 $G_{\mathrm{rem}}$  y G:width=35 $G_{\mathrm{rem}}$
\draw (244.500000,57.000000) -- (244.500000,0.000000);
\begin{scope}
\draw[fill=white] (244.500000, 47.500000) +(-45.000000:24.748737pt and 24.748737pt) -- +(45.000000:24.748737pt and 24.748737pt) -- +(135.000000:24.748737pt and 24.748737pt) -- +(225.000000:24.748737pt and 24.748737pt) -- cycle;
\clip (244.500000, 47.500000) +(-45.000000:24.748737pt and 24.748737pt) -- +(45.000000:24.748737pt and 24.748737pt) -- +(135.000000:24.748737pt and 24.748737pt) -- +(225.000000:24.748737pt and 24.748737pt) -- cycle;
\draw (244.500000, 47.500000) node {$G_{\mathrm{rem}}$};
\end{scope}
\begin{scope}
\draw[fill=white] (244.500000, 0.000000) +(-45.000000:24.748737pt and 11.313708pt) -- +(45.000000:24.748737pt and 11.313708pt) -- +(135.000000:24.748737pt and 11.313708pt) -- +(225.000000:24.748737pt and 11.313708pt) -- cycle;
\clip (244.500000, 0.000000) +(-45.000000:24.748737pt and 11.313708pt) -- +(45.000000:24.748737pt and 11.313708pt) -- +(135.000000:24.748737pt and 11.313708pt) -- +(225.000000:24.748737pt and 11.313708pt) -- cycle;
\draw (244.500000, 0.000000) node {$G_{\mathrm{rem}}$};
\end{scope}
% Done with gates; drawing ending labels
\draw[color=black] (272.000000,57.000000) node[right] {$\ket{u}$};
\draw[color=black] (272.000000,38.000000) node[right] {$\ket{v}$};
\draw[color=black] (272.000000,19.000000) node[right] {$\ket{0}$};
\draw[color=black] (272.000000,0.000000) node[right] {$\ket{y\oplus G(uv)}$};
% Done with ending labels; drawing cut lines and comments
% Done with comments
\end{tikzpicture}}
\caption{Support-peeling step in Case 2 of the large-support recursion. The circuit XORs $\bar{\tau}_G(u)$ into the suffix register, computes a flag indicating whether the modified suffix equals $0^{m-r}$, applies $H_G$ to contribute $G_{\mathrm{peel}}$, uncomputes the flag and suffix update, and then applies the recursive implementation of $G_{\mathrm{rem}}$. The two connected boxes labeled $H_G$ emphasize that this operation depends on $(z,u)$ while the suffix register is bypassed.}
\label{fig:sparse-large-support-case2}
\end{figure}
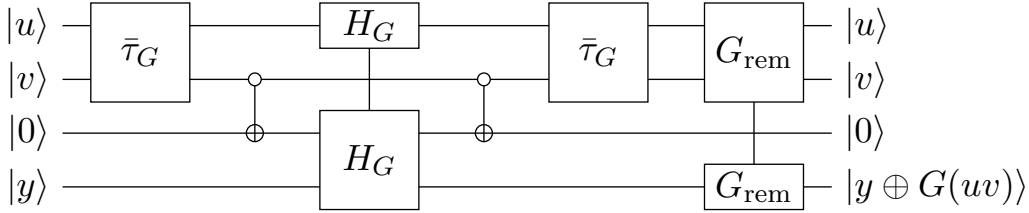

This circuit is correct:
\begin{itemize}
\item if $x=u\tau_G(u)\in A_G$, then the modified suffix becomes $0^{m-r}$, so $z=1$ and the output changes by $H_G(1,u)=G_{\mathrm{peel}}(u\tau_G(u))=G_{\mathrm{peel}}(x)$;
\item if $u\in U_G$ but $x\notin A_G$, then $v\neq\tau_G(u)$, so the modified suffix is nonzero and hence $z=0$;
\item if $u\notin U_G$, then $\bar{\tau}_G(u)=0^{m-r}$ and $H_G(z,u)=0^b$ for both values of $z$.
\end{itemize}
Thus the circuit contributes exactly $G_{\mathrm{peel}}$, and composing it with the recursive circuit for $G_{\mathrm{rem}}$ yields a circuit for $G$.

The non-recursive work in Case 2 consists of synthesizing $\bar{\tau}_G$, synthesizing $H_G$, and computing the flag on the suffix register. By \cref{lem:total-boolean}, $\bar{\tau}_G$ costs $O(\sqrt{(m-r)\,2^r})$ and $H_G$ costs $O(\sqrt{b\,2^{r+1}})$.
The flag computation is an $(m-r)$-controlled Toffoli with open controls on the suffix bits, which can be decomposed into $O(m-r)=O(m)$ Toffoli gates using $O(1)$ ancillas~\cite{khattar2025rise}. Hence the total non-recursive Toffoli cost of Case 2 is
$O(\sqrt{(m-r)\,2^r}+\sqrt{b\,2^r}+m)=O(\sqrt{\frac{ns}{c}}+\sqrt{\frac{bs}{c}}+n).$

\paragraph{Progress and termination.}
Case 1 sends the state $(m,t)$ to $(m-1,t)$.
Case 2 keeps $m$ fixed and decreases the support size by at least $p_G\ge 2^{r-1}\ge \frac{s}{2c}$.
Therefore Case 2 can occur at most $\lceil \frac{s}{2^{r-1}}\rceil = O(c)$ times along any recursion path.
Similarly, Case 1 can occur at most $n-r\le n$ times because each such step decreases $m$ by one.
The recursion terminates when the current support size drops to at most $n^{3/2}\log n$, in which case we apply \cref{lem:sparse-supp-boolean}, or when the input length reaches $m=r$, in which case we apply \cref{lem:total-boolean}.

\paragraph{Toffoli count.}
Each Case-1 step contributes $O\!\left(\sqrt{\frac{s}{c}\log\frac{s}{c}}\right)$, so all Case-1 steps together contribute $O\!\left(n\sqrt{\frac{s}{c}\log\frac{s}{c}}\right)$.
Each Case-2 step contributes $O\!\left(\sqrt{\frac{ns}{c}}+\sqrt{\frac{bs}{c}}+n\right)$, and there are $O(c)$ such steps, so all Case-2 steps together contribute
$O(\sqrt{nsc}+\sqrt{bsc}+nc)$.

There are two possible leaves.
The small-support leaf contributes $O(n^{3/2}\log n)$, and the leaf at $m=r$ contributes $O(\sqrt{bs/c})$.
The latter is absorbed by $O(\sqrt{bsc})$ because $c\ge 1$.
Hence the total Toffoli count satisfies
$O(n\sqrt{\frac{s}{c}\log\frac{s}{c}}+\sqrt{nsc}+\sqrt{bsc}+nc+n^{3/2}\log n)$.
Now set $c:=\lceil \sqrt{n}\rceil$. Then, the total Toffoli count is
$O(n^{3/4}\sqrt{s\log\frac{s}{\sqrt{n}}}+n^{1/4}\sqrt{bs})$.

\paragraph{Ancillas.}
We separate the ancilla count into persistent registers and reusable scratch space.
With the choice $c=\lceil\sqrt{n}\rceil$, the persistent registers come only from Case 1: in each prefix-compression step, the newly computed $(r-1)$-bit string $\bar{\gamma}_G(u)$ becomes part of the logical input to the recursive call and must therefore remain available until that call returns.
Since there are at most $n-r$ Case-1 steps, these persistent registers contribute $O((n-r)(r-1))=O\!\left(n\log\frac{s}{\sqrt{n}}\right)$ ancillas.

All remaining work bits are reusable.
The temporary ancillas used to synthesize $\bar{\gamma}_G$, $\bar{\tau}_G$, and $H_G$ are uncomputed before the next recursive step starts, and the same is true for the ancillas used to compute the suffix flag in Case 2.
The terminal call to \cref{lem:sparse-supp-boolean}, when it occurs, uses only $O(\log n)$ ancillas and can also share this scratch space.
These Boolean-function subroutines are instantiated using \cref{lem:total-boolean}, so the reusable part is bounded by
\begin{align}
O\!\left(\max\left\{\sqrt{(r-1)2^r},\sqrt{(n-r)2^r},\sqrt{b2^r}\right\}+\log n\right)
=O\!\left(n^{1/4}\sqrt{s}+n^{-1/4}\sqrt{bs}+\log n\right),
\end{align}
where the final $O(\log n)$ term comes from the possible terminal call to \cref{lem:sparse-supp-boolean}; the suffix-flag computation in Case 2 contributes only $O(1)$ reusable scratch space~\cite{khattar2025rise}.

Consequently, the recursive framework uses $O\!\left(n\log\frac{s}{\sqrt{n}}+n^{1/4}\sqrt{s}+n^{-1/4}\sqrt{bs}\right)$ ancillas. 

\end{proof}

\section{Lower Bounds}
\label{sec:lower-bound}
In this section we prove lower bounds on the worst-case $T$-count for approximately preparing sparse quantum states.
For an $n$-qubit pure state $\ket{\psi}$, let $\mathcal T_\epsilon(\ket{\psi})$ denote the minimum number of $T$ gates in a unitary Clifford+$T$ circuit that, starting from $\ket{0^{n+a}}$ for some $a\ge 0$, prepares an $(n+a)$-qubit pure state $\ket{\phi}$ with
$D_{\mathrm{tr}}(\ket{\phi},\ket{\psi}\otimes \ket{0^a})\le \epsilon$.
Define
\begin{align}
\mathcal T_\epsilon^{\mathrm{spar}}(n,s)
=
\max_{\ket{\psi}:\,|\operatorname{supp}(\ket{\psi})|\le s}
\mathcal T_\epsilon(\ket{\psi}).
\end{align}

The main theorem in this section is as follows:
\begin{theorem}
\label{thm:sparse-state-lower-bound}
For every $2\le s\le 2^n$ and every $0<\epsilon<1$, there exists an $n$-qubit $s$-sparse state $\ket{\psi}$ such that
$\mathcal T_\epsilon(\ket{\psi})
=
\Omega(\sqrt{s\log(1/\epsilon)}+\log(1/\epsilon))$.
Moreover, for every $0<\epsilon\le \frac{1}{6}$ and every $2\le s\le 2^{n/2}$, there exists an $n$-qubit $s$-sparse state $\ket{\psi}$ such that
$\mathcal T_\epsilon(\ket{\psi})
=
\Omega(\min\{s,\sqrt{ns}\})$.
\end{theorem}

\begin{remark}
For simplicity, we define $\mathcal T_\epsilon$ using unitary Clifford+$T$ circuits.
The same asymptotic lower bounds also hold in the stronger adaptive Clifford+$T$ model,
where mid-circuit measurements, classical feed-forward, and mixed-state outputs are allowed,
after replacing $\epsilon$ by $\sqrt{6\epsilon}$, see \cite[Claim~4.5]{gosset2024quantum}. 
% Indeed, by \cite[Claim~4.5]{gosset2024quantum}, any such adaptive protocol with expected $T$-count $t$ and error $\epsilon$ yields a Clifford computation with Pauli postselection and $O(t)$ magic states that prepares a pure state within trace distance $O(\sqrt{\epsilon})$ of the target. The first bound then follows from the same counting argument, while the second follows by combining the counting argument with the $W_s$-state lower bound.
\end{remark}

\begin{proof}
For the first claim, let $m:=\lfloor \log_2 s\rfloor$. By \cite[Theorem~4.1]{gosset2024quantum}, there exists an $m$-qubit state $\ket{\phi}$ such that
\begin{align}
\mathcal T_\epsilon(\ket{\phi})
=
\Omega(\sqrt{2^m\log(1/\epsilon)}+\log(1/\epsilon))
=
\Omega(\sqrt{s\log(1/\epsilon)}+\log(1/\epsilon)).
\end{align}
Viewing $\ket{\phi}$ as an $n$-qubit state by appending $\ket{0^{n-m}}$ gives an $s$-sparse witness.

For the second claim, we split into the regimes $s=O(n)$ and $s=\Omega(n)$. The case $s=\Omega(n)$ with $s\le 2^{n/2}$ is proved in \cref{sebsec:large-s-regime}, see \cref{lem:sparse-counting-lower-bound}. The case $s=O(n)$ is proved in \cref{sebsec:small-s-regime}, see \cref{lem:w-state-lower-bound}. 
\end{proof}

\subsection{Large-$s$ regime}
\label{sebsec:large-s-regime}
% By \cite{gosset2024quantum}, some $n$-qubit state requires $\Omega(\sqrt{2^n\log 1/\epsilon}+\log 1/\epsilon)$ $T$ gates to prepare. This immediately yields the lower bound $\Omega(\sqrt{s\log 1/\epsilon}+\log 1/\epsilon)$ for $n$-qubit $s$-sparse states. 

We combine a packing argument with the following counting lemma.

\begin{lemma}[\cite{gosset2024quantum}]
\label{lem:unitary-counting}
Let $0<\epsilon<1/2$, and let $\mathcal F$ be a family of $n$-qubit pure states such that
$D_{\mathrm{tr}}(\ket{\psi},\ket{\phi})>2\epsilon$
for all distinct $\ket{\psi},\ket{\phi}\in \mathcal F$. Suppose every state in $\mathcal F$ admits an $\epsilon$-approximate preparation by a unitary Clifford+$T$ circuit with at most $t$ $T$ gates. Then $|\mathcal F|\le 2^{O(n^2+t^2)}$.
\end{lemma}

We now construct a large packing of sparse states.

\begin{lemma}
\label{lem:sparse-uniform-packing}
Let $0<\epsilon\le 1/2$, and write $N:=2^n$.
Then there exists a family $\mathcal F$ of $n$-qubit $s$-sparse states such that every two distinct states in $\mathcal F$ have trace distance greater than $\epsilon$, and
\begin{align}
|\mathcal F|
\ge
\frac{1}{s+1}
(2e)^{-s/2}
\left(\frac{N}{s}\right)^{(1-2\epsilon^2)s}.
\end{align}
\end{lemma}

\begin{proof}
For each $s$-subset $S\subseteq [N]$, let $\ket{u_S}:=\frac{1}{\sqrt s}\sum_{j\in S}\ket{j}$ and $\mathcal U:=\{\ket{u_S}: S\subseteq [N],\ |S|=s\}$.
Set $r:=\lfloor s(1-\sqrt{1-\epsilon^2})\rfloor$. For $S,T\subseteq [N]$ with $|S|=|T|=s$, we have $\braket{u_S|u_T}=|S\cap T|/s$, and hence
\begin{align}
D_{\mathrm{tr}}(\ket{u_S},\ket{u_T})\le \epsilon
\qquad \Longleftrightarrow \qquad
s-|S\cap T|\le r.
\end{align}
Fix $S$. The number of sets $T$ such that $D_{\mathrm{tr}}(\ket{u_S},\ket{u_T})\le \epsilon$ is
\begin{align}
N_\epsilon
:=
\sum_{i=0}^{r}\binom{s}{i}\binom{N-s}{i}.
\end{align}

For each $0\le i\le r$, $\binom{s}{i}\binom{N-s}{i}\le \binom{N}{2i}$, since choosing $i$ elements of $S$ and $i$ elements of $[N]\setminus S$ determines a $2i$-subset of $[N]$. Since $\epsilon\le 1/2$, we have $r\le \epsilon^2 s\le s/4$, so $2r\le s/2\le N/2$. Hence the binomial coefficients are increasing on $\{0,\dots,2r\}$, and
\begin{align}
N_\epsilon
\le
\sum_{i=0}^{r}\binom{N}{2i}
\le
(r+1)\binom{N}{2r}.
\end{align}

We now greedily construct a family $\mathcal F\subseteq \mathcal U$ whose elements are pairwise more than $\epsilon$ apart. Start with $\mathcal V:=\mathcal U$ and $\mathcal F:=\emptyset$. While $\mathcal V\neq \emptyset$, choose any state $\ket{u_S}\in \mathcal V$, add it to $\mathcal F$, and delete from $\mathcal V$ every state $\ket{u_T}$ satisfying
\begin{align}
D_{\mathrm{tr}}(\ket{u_S},\ket{u_T})\le \epsilon.
\end{align}
At each step, the number of deleted states is at most $N_\epsilon$, so
\begin{align}
|\mathcal F|
\ge
\frac{|\mathcal U|}{N_\epsilon}
=
\frac{\binom{N}{s}}{N_\epsilon}.
\end{align}
By construction, every two distinct states in $\mathcal F$ have trace distance greater than $\epsilon$.

Using $\binom{N}{s}\ge (N/s)^s$ and $\binom{N}{2r}\le (eN/2r)^{2r}$, we obtain
\begin{align}
\frac{\binom{N}{s}}{N_\epsilon}
\ge
\frac{1}{r+1}
\left(\frac{2r}{es}\right)^{2r}
\left(\frac{N}{s}\right)^{s-2r}.
\end{align}
Because $2r/s\le 1/2$ and the function $x\mapsto (x/e)^x$ is decreasing on $(0,1]$,
\begin{align}
\left(\frac{2r}{es}\right)^{2r}
=
\left[\left(\frac{2r/s}{e}\right)^{2r/s}\right]^s
\ge
\left(\frac{1}{\sqrt{2e}}\right)^s.
\end{align}
Since $r+1\le s+1$ and $s-2r\ge (2\sqrt{1-\epsilon^2}-1)s\ge (1-2\epsilon^2)s$,
\begin{align}
\frac{\binom{N}{s}}{N_\epsilon}
\ge
\frac{1}{s+1}
(2e)^{-s/2}
\left(\frac{N}{s}\right)^{(1-2\epsilon^2)s}.
\end{align}
This proves the claim.
\end{proof}

Combining the packing lemma with the counting lemma now gives the desired lower bound in the large-$s$ regime.

\begin{lemma}
\label{lem:sparse-counting-lower-bound}
There exists an absolute constant $C>0$ such that for every $Cn \le s \le 2^{n/2}$ and $0<\epsilon\le \frac{1}{6}$,
$\mathcal T_\epsilon^{\mathrm{spar}}(n,s)=\Omega(\sqrt{ns})$.
\end{lemma}

\begin{proof}
Set $\delta_0:=3\epsilon$ and $N:=2^n$. By \cref{lem:sparse-uniform-packing}, there exists a family $\mathcal F$ of $n$-qubit $s$-sparse states such that every two distinct states in $\mathcal F$ have trace distance greater than $\delta_0$, and
\begin{align}
\label{eq:sparse-packing-size}
|\mathcal F|
\ge
\frac{1}{s+1}
(2e)^{-s/2}
\left(\frac{N}{s}\right)^{(1-2\delta_0^2)s}.
\end{align}

Set $t:=\mathcal T_{\epsilon}^{\mathrm{spar}}(n,s)$. By definition of $t$, every state in $\mathcal F$ admits an $\epsilon$-approximate preparation using at most $t$ $T$ gates. Since $\delta_0=3\epsilon>2\epsilon$, \cref{lem:unitary-counting} gives
\begin{align}
2^{O(n^2+t^2)}
\ge
|\mathcal F|.
\end{align}
Taking logarithms gives
\begin{align}
n^2+t^2
=
\Omega\!\left((1-2\delta_0^2)s(n-\log s)-\frac{s}{2}\log 2e-\log (s+1)\right).
\end{align}
If $s\le 2^{n/2}$, then the right-hand side is $\Omega(sn)$. Choosing the absolute constant $C$ in the statement large enough, it follows that $t=\Omega(\sqrt{ns})$ whenever $s\ge Cn$.
\end{proof}

\subsection{Small-$s$ regime}
\label{sebsec:small-s-regime}
The counting argument weakens when $s=O(n)$, because the canonical-form bound contributes an additive $n^2$ term in the exponent. For this regime we exploit the structure of $W$ states.

Let $\mathcal P_m$ denote the Pauli group on $m$ qubits. For an $m$-qubit pure state $\ket{\phi}$, define $\stab(\ket{\phi}):=\{g\in \mathcal P_m: g\ket{\phi}=\ket{\phi}\}$. For a stabilizer subgroup $A\le \mathcal P_m$, that is, an abelian subgroup with $-I\notin A$, let $\rank(A):=\log_2|A|$ and $\Pi_A:=\frac{1}{|A|}\sum_{g\in A}g$. Then $\Pi_A$ is the orthogonal projector onto the common $+1$ eigenspace of $A$.

For an $n$-qubit pure state $\ket{\psi}$, define its stabilizer nullity by $\nu(\ket{\psi}):=n-\rank(\stab(\ket{\psi}))$. The exact $T$-count of $\ket{\psi}$ is at least $\nu(\ket{\psi})$~\cite{beverland2020lower}. Now consider $\ket{W_s}\otimes \ket{0^{n-s}}$, where $\ket{W_s}:=\frac{1}{\sqrt{s}}\sum_{j=1}^s \ket{e_j}$ and $\ket{e_j}$ denotes the computational basis state of Hamming weight $1$ whose unique $1$ occurs in position $j$. Since $\nu(\ket{W_s}\otimes \ket{0^{n-s}})=s-1$, exact preparation requires at least $s-1$ $T$ gates~\cite{vilmart2025resource}. We show that a linear lower bound survives approximation: $\epsilon$-approximate preparation still requires $\Omega((1-\epsilon^2)s)$ $T$ gates.

We begin with a reformulation of approximate preparation as an optimization problem over stabilizer projectors.

\begin{lemma}
\label{lem:rank-under-trace-distance}
For every $m$-qubit pure state $\ket{\psi}$ and every $0\le \epsilon<1$,
\begin{align}
\max_{\ket{\phi}:\,D_{\mathrm{tr}}(\ket{\psi},\ket{\phi})\le \epsilon}
\rank(\stab(\ket{\phi}))
=
\max_{A:\,\|\Pi_A\ket{\psi}\|_2^2\ge 1-\epsilon^2}
\rank(A).
\end{align}
\end{lemma}

\begin{proof}
Write $L$ and $R$ for the left-hand side and the right-hand side.

To prove $L\le R$, fix $\ket{\phi}$ with $D_{\mathrm{tr}}(\ket{\psi},\ket{\phi})\le \epsilon$ and set $A:=\stab(\ket{\phi})$. Since $\ket{\phi}\in \operatorname{im}(\Pi_A)$, the normalized projection of $\ket{\psi}$ onto $\operatorname{im}(\Pi_A)$, denoted by $\ket{\xi}:=\Pi_A\ket{\psi}/\|\Pi_A\ket{\psi}\|_2$, is the unit vector in that subspace with maximum overlap with $\ket{\psi}$. Hence
\begin{align}
D_{\mathrm{tr}}(\ket{\psi},\ket{\xi})
\le
D_{\mathrm{tr}}(\ket{\psi},\ket{\phi})
\le
\epsilon.
\end{align}
Using the trace-distance formula for pure states gives $\|\Pi_A\ket{\psi}\|_2^2\ge 1-\epsilon^2$, so $A$ is feasible for $R$ and therefore $\rank(\stab(\ket{\phi}))=\rank(A)\le R$. Taking the maximum over $\ket{\phi}$ yields $L\le R$.

For the reverse inequality, fix any stabilizer subgroup $A$ with $\|\Pi_A\ket{\psi}\|_2^2\ge 1-\epsilon^2$, and define $\ket{\chi}:=\Pi_A\ket{\psi}/\|\Pi_A\ket{\psi}\|_2$. Then
\begin{align}
D_{\mathrm{tr}}(\ket{\psi},\ket{\chi})
=
\sqrt{1-\|\Pi_A\ket{\psi}\|_2^2}
\le
\epsilon.
\end{align}
Every element of $A$ fixes every vector in $\operatorname{im}(\Pi_A)$, so $A\subseteq \stab(\ket{\chi})$. Hence $\rank(A)\le \rank(\stab(\ket{\chi})) \le L$. Taking the maximum over feasible $A$ gives $R\le L$.
\end{proof}

The next lemma shows that any stabilizer subgroup having noticeable overlap with the $W$ state must have bounded rank.

\begin{lemma}
\label{lem:w-state-large-stabilizer-subgroup}
Let $c\in(0,1]$. If $\frac{4}{c^2}\le s\le n$ and $A\le \mathcal P_n$ satisfies $\|\Pi_A(\ket{W_s}\otimes \ket{0^{n-s}})\|_2\ge c$, then
\begin{align}
\rank(A)
\le
n-c^2s+1+\log_2\frac{2}{c^2}.
\end{align}
\end{lemma}

\begin{proof}
Set $\ket{\psi}:=\ket{W_s}\otimes \ket{0^{n-s}}$, let
\begin{align}
\mathcal Z_n
:=
\left\{
\pm Z_1^{a_1}\cdots Z_n^{a_n} : a_1,\dots,a_n\in\{0,1\}
\right\},
\end{align}
and write $D:=A\cap \mathcal Z_n$, $d:=\rank(D)$, and $t:=\rank(A)-\rank(D)$. Since the common $+1$ eigenspace of $A$ is contained in that of $D$, we have $\|\Pi_D\ket{\psi}\|_2\ge \|\Pi_A\ket{\psi}\|_2\ge c$.

Because $D$ is diagonal, $\Pi_D$ acts diagonally in the computational basis. Let $r$ be the number of states among $\ket{e_1}\otimes \ket{0^{n-s}},\dots,\ket{e_s}\otimes \ket{0^{n-s}}$ fixed by every element of $D$. Then $\|\Pi_D\ket{\psi}\|_2^2=r/s$, so $r\ge c^2s$. After relabeling the first $s$ qubits if necessary, we may assume that the surviving basis states are $\ket{e_1}\otimes \ket{0^{n-s}},\dots,\ket{e_r}\otimes \ket{0^{n-s}}$.

A diagonal Pauli $g=\varepsilon \prod_{j=1}^n Z_j^{a_j}$ fixes $\ket{e_1}\otimes \ket{0^{n-s}},\dots,\ket{e_r}\otimes \ket{0^{n-s}}$ iff $\varepsilon(-1)^{a_1}=\cdots=\varepsilon(-1)^{a_r}=1$, equivalently $a_1=\cdots=a_r$ and $\varepsilon=(-1)^{a_1}$. Thus the diagonal Paulis that fix these states form a subgroup of rank $n-r+1$, generated by $-Z_1\cdots Z_r,Z_{r+1},\dots,Z_n$. Since $D$ is contained in this subgroup,
\begin{align}
d\le n-r+1\le n-c^2s+1.
\end{align}

We next bound the contribution of non-diagonal Paulis. Let $g\in \mathcal P_n\setminus \mathcal Z_n$, and let $F$ be the set of qubits on which $g$ acts by $X$ or $Y$. If $F$ intersects $\{s+1,\dots,n\}$, then $g(\ket{e_j}\otimes \ket{0^{n-s}})$ lies outside the support of $\ket{\psi}$ for every $j\in[s]$. Otherwise $F\subseteq [s]$. For a weight-$1$ basis state $\ket{e_j}$, the first $s$ qubits of $g(\ket{e_j}\otimes \ket{0^{n-s}})$ have Hamming weight $|F|+1$ when $j\notin F$ and $|F|-1$ when $j\in F$. Thus it can lie in the support of $\ket{\psi}$ only when $|F|=2$ and $j\in F$, so at most two basis states are mapped back into the support. Therefore
\begin{align}
\label{eq:non-diagonal-w-overlap}
\left|\bra{\psi}g\ket{\psi}\right|
\le
\frac{2}{s}.
\end{align}

It remains to bound the full projector overlap. Expanding $\|\Pi_A\ket{\psi}\|_2^2$ gives
\begin{align}
\|\Pi_A\ket{\psi}\|_2^2
&=
\bra{\psi}\Pi_A\ket{\psi} \\
&=
\frac{1}{|A|}\sum_{g\in D}\bra{\psi}g\ket{\psi}
+
\frac{1}{|A|}\sum_{g\in A\setminus D}\bra{\psi}g\ket{\psi}.
\end{align}
Since $|A|=2^{d+t}$ and $|D|=2^d$, the contribution of $D$ equals
\begin{align}
\frac{1}{|A|}\sum_{g\in D}\bra{\psi}g\ket{\psi}
=
2^{-t}\bra{\psi}\Pi_D\ket{\psi}
=
2^{-t}\|\Pi_D\ket{\psi}\|_2^2.
\end{align}
Using \cref{eq:non-diagonal-w-overlap}, we obtain
\begin{align}
c^2
&\le
\|\Pi_A\ket{\psi}\|_2^2 \\
&\le
2^{-t}\|\Pi_D\ket{\psi}\|_2^2
+
\frac{1}{|A|}\sum_{g\in A\setminus D}\left|\bra{\psi}g\ket{\psi}\right| \\
&\le
2^{-t}\|\Pi_D\ket{\psi}\|_2^2
+
\frac{|A|-|D|}{|A|}\cdot \frac{2}{s} \\
&\le
2^{-t}+\frac{2}{s}.
\end{align}
Since $s\ge 4/c^2$, we have $2/s\le c^2/2$, and therefore $2^{-t}\ge c^2/2$, or equivalently $t\le \log_2(2/c^2)$. Thus, $\rank(A)=d+t\le n-c^2s+1+\log_2\frac{2}{c^2}$.
\end{proof}

We can now combine the previous two lemmas to obtain the approximate-preparation lower bound for $W$ states.

\begin{lemma}
\label{lem:w-state-lower-bound}
Assume $0\le \epsilon\le \frac{1}{2}$ and $\frac{4}{1-\epsilon^2}\le s\le n$.
\begin{align}
\mathcal T_\epsilon(\ket{W_s}\otimes \ket{0^{n-s}})
=
\Omega((1-\epsilon^2)s).
\end{align}
\end{lemma}

\begin{proof}
Let $\ket{\psi}:=\ket{W_s}\otimes \ket{0^{n-s}}$, and let $t:=\mathcal T_\epsilon(\ket{\psi})$. By definition of $t$, there exist an integer $a\ge 0$ and an $(n+a)$-qubit pure state $\ket{\phi}$ prepared by a unitary Clifford+$T$ circuit with $t$ $T$ gates such that
\begin{align}
D_{\mathrm{tr}}(\ket{\phi},\ket{\psi}\otimes \ket{0^a})\le \epsilon.
\end{align}
Set $\ket{\psi'}:=\ket{\psi}\otimes \ket{0^a}=\ket{W_s}\otimes \ket{0^{n-s+a}}$. Since $\frac{4}{1-\epsilon^2}\le s\le n\le n+a$, \cref{lem:w-state-large-stabilizer-subgroup} applies to $\ket{\psi'}$. The exact $T$-count of a pure state is at least its stabilizer nullity. Therefore, by \cref{lem:rank-under-trace-distance} and \cref{lem:w-state-large-stabilizer-subgroup},
\begin{align}
t
&\ge
\min_{\ket{\varphi}:\,D_{\mathrm{tr}}(\ket{\psi'},\ket{\varphi})\le \epsilon}
\Bigl((n+a)-\rank(\stab(\ket{\varphi}))\Bigr) \\
&=
(n+a)-
\max_{\ket{\varphi}:\,D_{\mathrm{tr}}(\ket{\psi'},\ket{\varphi})\le \epsilon}
\rank(\stab(\ket{\varphi})) \\
&=
(n+a)-
\max_{A:\,\|\Pi_A\ket{\psi'}\|_2^2\ge 1-\epsilon^2}
\rank(A) \\
&\ge
(1-\epsilon^2)s-1-\log_2\frac{2}{1-\epsilon^2}
=
\Omega((1-\epsilon^2)s).
\end{align}
\end{proof}

\section{Conclusion}
\label{sec:conclusion-framework}

We studied sparse quantum state preparation in the fault-tolerant Clifford+$T$ model. Our main upper bound shows that any $n$-qubit $s$-sparse state can be prepared up to constant error using $\widetilde{O}(\min\{s,\ n^{3/4}\sqrt{s}\})$ $T$ gates. Thus, while the previous linear $O(s)$ bound remains optimal in the small-support regime, it is not inherent once the support becomes sufficiently large. The key technical ingredient is a support-aware synthesis theorem for sparse Boolean functions, which may be of independent interest beyond the state-preparation setting.

We also established lower bounds of $\Omega(\min\{s,\sqrt{ns}\})$ for every $0<\epsilon\le 1/6$ and $2\le s\le 2^{n/2}$. In particular, the linear upper bound is tight when $s=O(n)$, while for $\widetilde{\Omega}(n^{3/2}) \le s \le 2^{n/2}$ our upper and lower bounds differ by at most a factor of $\widetilde{O}(n^{1/4})$. An important open question is to close this remaining gap, and more specifically to determine whether the crossover around $s=n^{3/2}$ up to polylogarithmic factors and the factor $n^{3/4}$ arise from the present algorithm or reflect inherent limitations of fault-tolerant sparse quantum state preparation.

\bibliographystyle{alpha}
\bibliography{refs}

@article{babbush2018encoding,
  title={Encoding electronic spectra in quantum circuits with linear T complexity},
  author={Babbush, Ryan and Gidney, Craig and Berry, Dominic W and Wiebe, Nathan and McClean, Jarrod and Paler, Alexandru and Fowler, Austin and Neven, Hartmut},
  journal={Physical Review X},
  volume={8},
  number={4},
  pages={041015},
  year={2018},
  publisher={APS}
}

@article{vilmart2025resource,
  title={Resource-Efficient Synthesis of Sparse Quantum States},
  author={Vilmart, Renaud and Ty, Sunheang and Mang, Chetra},
  journal={arXiv preprint arXiv:2508.05386},
  year={2025}
}

@article{rupprecht2026sparse,
  title={Sparse quantum state preparation with improved Toffoli cost},
  author={Rupprecht, Felix and W{\"o}lk, Sabine},
  journal={arXiv preprint arXiv:2601.09388},
  year={2026}
}

@article{litinski2019game,
  title={A game of surface codes: Large-scale quantum computing with lattice surgery},
  author={Litinski, Daniel},
  journal={Quantum},
  volume={3},
  pages={128},
  year={2019},
  publisher={Verein zur F{\"o}rderung des Open Access Publizierens in den Quantenwissenschaften}
}

@article{beverland2020lower,
  title={Lower bounds on the non-Clifford resources for quantum computations},
  author={Beverland, Michael and Campbell, Earl and Howard, Mark and Kliuchnikov, Vadym},
  journal={Quantum Science \& Technology},
  volume={5},
  number={3},
  pages={035009},
  year={2020},
  publisher={IOP Publishing}
}

@article{wang2025faster,
  title={Faster State Preparation with Randomization},
  author={Wang, Yue and Zhang, Xiao-Ming and Yuan, Xiao and Zhao, Qi},
  journal={arXiv e-prints},
  pages={arXiv--2510},
  year={2025}
}

@article{harrow2025randomized,
  title={Randomized truncation of quantum states},
  author={Harrow, Aram W and Lowe, Angus and Witteveen, Freek},
  journal={arXiv preprint arXiv:2510.08518},
  year={2025}
}

@book{nielsen2010quantum,
  title={Quantum computation and quantum information},
  author={Nielsen, Michael A and Chuang, Isaac L},
  year={2010},
  publisher={Cambridge university press}
}

@article{bravyi2012magic,
  title={Magic-state distillation with low overhead},
  author={Bravyi, Sergey and Haah, Jeongwan},
  journal={Physical Review A—Atomic, Molecular, and Optical Physics},
  volume={86},
  number={5},
  pages={052329},
  year={2012},
  publisher={APS}
}

@article{khattar2025rise,
  title={Rise of conditionally clean ancillae for efficient quantum circuit constructions},
  author={Khattar, Tanuj and Gidney, Craig},
  journal={Quantum},
  volume={9},
  pages={1752},
  year={2025},
  publisher={Verein zur F{\"o}rderung des Open Access Publizierens in den Quantenwissenschaften}
}

@article{luo2025space,
  title={Space-time tradeoff for sparse quantum state preparation},
  author={Luo, Jingquan and Li, Guanzhong and Li, Lvzhou},
  journal={arXiv preprint arXiv:2506.16964},
  year={2025}
}

@inproceedings{gosset2024quantum,
  title={Quantum state preparation with optimal T-count},
  author={Gosset, David and Kothari, Robin and Wu, Kewen},
  booktitle={Proceedings of the 2026 Annual ACM-SIAM Symposium on Discrete Algorithms (SODA)},
  pages={3378--3406},
  year={2026},
  organization={SIAM}
}

@article{low2024trading,
  title={Trading T gates for dirty qubits in state preparation and unitary synthesis},
  author={Low, Guang Hao and Kliuchnikov, Vadym and Schaeffer, Luke},
  journal={Quantum},
  volume={8},
  pages={1375},
  year={2024},
  publisher={Verein zur F{\"o}rderung des Open Access Publizierens in den Quantenwissenschaften}
}

@article{zhang2024circuit,
  title={Circuit complexity of quantum access models for encoding classical data},
  author={Zhang, Xiao-Ming and Yuan, Xiao},
  journal={npj Quantum Information},
  volume={10},
  number={1},
  pages={42},
  year={2024},
  publisher={Nature Publishing Group UK London}
}

@InProceedings{luo2024circuit,
  author =	{Li, Lvzhou and Luo, Jingquan},
  title =	{{Nearly Optimal Circuit Size for Sparse Quantum State Preparation}},
  booktitle =	{52nd International Colloquium on Automata, Languages, and Programming (ICALP 2025)},
  pages =	{113:1--113:19},
  series =	{Leibniz International Proceedings in Informatics (LIPIcs)},
  ISBN =	{978-3-95977-372-0},
  ISSN =	{1868-8969},
  year =	{2025},
  volume =	{334},
  editor =	{Censor-Hillel, Keren and Grandoni, Fabrizio and Ouaknine, Jo\"{e}l and Puppis, Gabriele},
  publisher =	{Schloss Dagstuhl -- Leibniz-Zentrum f{\"u}r Informatik},
  address =	{Dagstuhl, Germany},
  URL =		{https://drops.dagstuhl.de/entities/document/10.4230/LIPIcs.ICALP.2025.113},
  URN =		{urn:nbn:de:0030-drops-234900},
  doi =		{10.4230/LIPIcs.ICALP.2025.113}
}

@article{mozafari2022efficient,
  title={Efficient deterministic preparation of quantum states using decision diagrams},
  author={Mozafari, Fereshte and De Micheli, Giovanni and Yang, Yuxiang},
  journal={Physical Review A},
  volume={106},
  number={2},
  pages={022617},
  year={2022},
  publisher={APS}
}

@article{de2022double,
  title={Double sparse quantum state preparation},
  author={de Veras, Tiago ML and da Silva, Leon D and da Silva, Adenilton J},
  journal={Quantum Information Processing},
  volume={21},
  number={6},
  pages={204},
  year={2022},
  publisher={Springer}
}

@article{mao2024towards,
  title={Toward optimal circuit size for sparse quantum state preparation},
  author={Mao, Rui and Tian, Guojing and Sun, Xiaoming},
  journal={Physical Review A},
  volume={110},
  number={3},
  pages={032439},
  year={2024},
  publisher={APS}
}

@article{ramacciotti2023simple,
  title={Simple quantum algorithm to efficiently prepare sparse states},
  author={Ramacciotti, Debora and Lefterovici, Andreea I and Rotundo, Antonio F},
  journal={Physical Review A},
  volume={110},
  number={3},
  pages={032609},
  year={2024},
  publisher={APS}
}

@article{harrow2009quantum,
  title={Quantum algorithm for linear systems of equations},
  author={Harrow, Aram W and Hassidim, Avinatan and Lloyd, Seth},
  journal={Physical Review Letters},
  volume={103},
  number={15},
  pages={150502},
  year={2009},
  publisher={APS}
}

@article{childs2017quantum,
  title={Quantum algorithm for systems of linear equations with exponentially improved dependence on precision},
  author={Childs, Andrew M and Kothari, Robin and Somma, Rolando D},
  journal={SIAM Journal on Computing},
  volume={46},
  number={6},
  pages={1920--1950},
  year={2017},
  publisher={SIAM}
}

@article{low2019hamiltonian,
  title={Hamiltonian simulation by qubitization},
  author={Low, Guang Hao and Chuang, Isaac L},
  journal={Quantum},
  volume={3},
  pages={163},
  year={2019},
  publisher={Verein zur F{\"o}rderung des Open Access Publizierens in den Quantenwissenschaften}
}

@article{kerenidis2019q,
  title={q-means: A quantum algorithm for unsupervised machine learning},
  author={Kerenidis, Iordanis and Landman, Jonas and Luongo, Alessandro and Prakash, Anupam},
  journal={Advances in Neural Information Processing Systems},
  volume={32},
  year={2019}
}

@inproceedings{shende2005synthesis,
  title={Synthesis of quantum logic circuits},
  author={Shende, Vivek V and Bullock, Stephen S and Markov, Igor L},
  booktitle={Proceedings of the 2005 Asia and South Pacific Design Automation Conference},
  pages={272--275},
  year={2005}
}

@article{plesch2011quantum,
  title={Quantum-state preparation with universal gate decompositions},
  author={Plesch, Martin and Brukner, {\v{C}}aslav},
  journal={Physical Review A},
  volume={83},
  number={3},
  pages={032302},
  year={2011},
  publisher={APS}
}

@article{iten2016quantum,
  title={Quantum circuits for isometries},
  author={Iten, Raban and Colbeck, Roger and Kukuljan, Ivan and Home, Jonathan and Christandl, Matthias},
  journal={Physical Review A},
  volume={93},
  number={3},
  pages={032318},
  year={2016},
  publisher={APS}
}

@inproceedings{gleinig2021efficient,
  title={An efficient algorithm for sparse quantum state preparation},
  author={Gleinig, Niels and Hoefler, Torsten},
  booktitle={Proceedings of the 58th ACM/IEEE Design Automation Conference},
  pages={433--438},
  year={2021},
  organization={IEEE}
}

@article{malvetti2021quantum,
  title={Quantum circuits for sparse isometries},
  author={Malvetti, Emanuel and Iten, Raban and Colbeck, Roger},
  journal={Quantum},
  volume={5},
  pages={412},
  year={2021},
  publisher={Verein zur F{\"o}rderung des Open Access Publizierens in den Quantenwissenschaften}
}

@article{de2020circuit,
  title={Circuit-based quantum random access memory for classical data with continuous amplitudes},
  author={de Veras, Tiago ML and De Araujo, Ismael CS and Park, Daniel K and da Silva, Adenilton J},
  journal={IEEE Transactions on Computers},
  volume={70},
  number={12},
  pages={2125--2135},
  year={2020},
  publisher={IEEE}
}

@article{zhang2022quantum,
  title={Quantum state preparation with optimal circuit depth: Implementations and applications},
  author={Zhang, Xiao-Ming and Li, Tongyang and Yuan, Xiao},
  journal={Physical Review Letters},
  volume={129},
  number={23},
  pages={230504},
  year={2022},
  publisher={APS}
}

% \newpage
% \appendices

\end{document}